\documentclass[11pt]{article}

\usepackage[margin=1in]{geometry}
\usepackage{times}
\usepackage{booktabs}
\usepackage{comment}
\usepackage{amssymb,amsmath,amsthm}
\usepackage{mathtools}
\usepackage{enumitem}
\usepackage{accents}
\usepackage[ruled,vlined,linesnumbered]{algorithm2e}
\usepackage{multirow}
\usepackage{mdframed}
\usepackage{datetime2}
\usepackage{subcaption}
\usepackage{nicefrac}
\usepackage{float}
\usepackage{hhline}
\usepackage{makecell}
\usepackage{wrapfig}
\usepackage{xspace}
\usepackage{anyfontsize}
\usepackage[colorlinks=true, citecolor=blue, linkcolor=blue, urlcolor=blue]{hyperref}
\usepackage{natbib}
\usepackage{setspace}

\newtheorem{assumption}{Assumption}

\newtheorem{proposition}{Proposition}
\newtheorem{remark}{Remark}
\theoremstyle{definition}
\newtheorem{defn}{Definition}
\newmdtheoremenv[font=\small,skipbelow=0.1pt,skipabove=11pt]{algoprocedure}{Procedure}
\newmdtheoremenv[font=\small,skipbelow=0.1pt,skipabove=11pt]{activityrule}{Activity Rule}
\newmdtheoremenv[font=\small,skipbelow=0.1pt,skipabove=11pt]{query}{Query Type}

\newcommand{\ConvergencePhase}{Convergence Phase\xspace}
\newcommand{\MRPAR}{MRPAR\xspace}
\newcommand{\IntervalReduction}{Interval Reduction\xspace}
\newcommand{\refMRPARActivityRule}{Activity Rule~\ref{act:mrpar}\xspace}
\newcommand{\refIntervalReductionActivityRule}{Activity Rule~\ref{act:convergence}\xspace}
\newcommand{\appref}[1]{Appendix \ref{#1}}

\SetCommentSty{mycommfont}
\SetKwComment{Comment}{$\triangleright$\hspace{-0.1cm} }{}
\SetAlFnt{\small}

\DeclareRobustCommand{\perturbed}[1]{\accentset{\circ}{#1}}
\setlist[description]{nosep,font=\normalfont\itshape}

\DeclareMathOperator*{\argmax}{arg\,max}
\newcommand{\Parameter}[1]{\textbf{#1}\\}

\title{\textbf{iMLCA: Machine Learning-powered Iterative Combinatorial Auctions with Interval Bidding}\thanks{Some of the ideas presented in this paper were also described in a one-page abstract published in the proceedings of the 22nd ACM Conference on Electronic Commerce (EC'21) \citep{BeyelerImlca2021}.}}

\author{
	Benjamin Lubin \\ Boston University \\ \texttt{blubin@bu.edu} 
	\and
	Manuel Beyeler \\ University of Zurich \\ \texttt{manuel.beyeler@uzh.ch}
	\and
	Gianluca Brero \\ Bryant University \\ \texttt{gbrero@bryant.edu}
	\and
	Sven Seuken \\ University of Zurich \\ \texttt{seuken@ifi.uzh.ch}
}

\date{}  

\begin{document}

\maketitle

\begin{abstract}
Preference elicitation is a major challenge in large combinatorial auctions because the bundle space grows exponentially in the number of items. Recent work has used machine learning (ML) algorithms to identify a small set of bundles to query from each bidder. However, a shortcoming of this prior work is that bidders must submit \textit{exact} values for the queried bundles, which can be quite costly. To address this, we propose \emph{iMLCA}, a new ML-powered iterative combinatorial auction with \textit{interval bidding} (i.e., where bidders submit upper and lower bounds instead of exact values). To steer the auction towards an efficient allocation, we introduce a price-based activity rule, asking bidders to tighten bounds on relevant bundles only. In our experiments, iMLCA achieves the same allocative efficiency as the prior ML-based auction that uses exact bidding. Moreover, it outperforms the well-known combinatorial clock auction in a realistically-sized domain.
\end{abstract}



\maketitle

\section{Introduction}

Combinatorial auctions (CAs) are used to allocate a set of heterogeneous items to multiple bidders who may see the items as substitutes or complements. In CAs, bidders can submit bid on \emph{bundles} rather than just individual items, which addresses problems such as \textit{exposure} and \textit{demand reduction} 
\citep{cramtonCAintro}.
CAs are widely used in real-world applications, including spectrum auctions \citep{cramton2013spectrum}, procurement \citep{sandholm2013very}, and TV ad allocation \citep{Goetzendorf:2015}.

Unfortunately, in CAs, the bundle space grows exponentially in the number of items, which makes it impractical for bidders to explore and report their full value function, even in medium-sized domains. Therefore, \emph{iterative combinatorial auctions (ICAs)} are often used in practice, where the auctioneer interacts with the bidders over the course of multiple rounds and only elicits a limited amount of information. A prime example is the \emph{combinatorial clock auction (CCA)}, which has been used in applications such as spectrum auctions \citep{cramton2013spectrum} and auctioning off the right to build offshore wind farms \citep{ausubel2011auction}, generating more than \$15 Billion in revenue since 2008 \citep{ausubel2017practical} .

To be practical, the CCA employs simple per-item prices, commonly referred to as \textit{linear prices}, during the first phase of the auction. Additionally, it restricts the number of bids that can be submitted in the second phase. For example, in the 2014 Canadian spectrum auction (with 98 goods), bidders could only submit up to 500 bids, which is a tiny fraction of the whole bundle space. But these simplifications are not without loss. As \citet{scheffel2012impact} and \citet{bichler2013core} have shown, bidders tend to focus on only a small number of bundles too early in the auction, which can lead to significant efficiency losses. This occurs primarily due to: 1) \textbf{omitted value:} bundles without submitted bids are disregarded in the final allocation, even if bidders assign them value; and 2) \textbf{coordination challenges:} bid compatibility between different bidders requires mutually exclusive item sets, a complex task when dealing with multiple items.

\subsection{ML-powered ICAs with Exact Bidding}

To address these problems with the CCA, \citet{brero2018combinatorial,brero2019machine} introduced \emph{MLCA}, a \emph{machine learning (ML)-powered iterative combinatorial auction}. MLCA drives the elicitation process via an \textit{ML-powered query module} (instead of using prices). This query module learns bidders’ value functions to identify which value queries to ask in each round. These learned value functions assign an inferred value for each possible bundle, thus mitigating the omitted value problems of the CCA. Furthermore, the value queries asked by MLCA are designed to ensure compatibility among bids, addressing the coordination challenges of the CCA. Empirical evidence presented by \citet{brero2019machine} confirms that the MLCA outperforms the CCA in achieving higher efficiency in a large CA domain. 

A shortcoming of MLCA is that it requires bidders to answer each query with an \textit{exact} value. In many applications, this is difficult for bidders, because determining the exact value of even a \emph{single} bundle typically involves a complex decision-making problem \citep{parkes2006}. For example, in spectrum auctions, a bundle of licenses corresponds to a mobile network operator being able to execute on a business plan (like offering a national cell plan). Thus, determining the value of a bundle of licenses is a multi-year profit-optimization problem, which requires time and resources to solve. Furthermore, submitting a bid close to a bidder's true value for a bundle (i.e., their maximum willingness-to-pay) often requires board approval, which involves a costly process \citep{Bichler2018PrincipalAgent}. In some auctions (e.g., for oil drilling rights), it may even be practically impossible to determine exact values, because there is inherent uncertainty about the value of the resources that cannot be resolved before the auction \citep{oren1975competitive}.
Requiring exact values is also costly from a computational perspective. While large-scale profit-optimization problems can be solved via mixed integer programming, this is a strongly NP-hard problem \citep{Garey1978StrongNPHardness}; thus, no FPTAS exists. However, anytime primal-dual algorithms that produce upper and lower bounds are available, indicating that computing bounds is generally easier than determining exact values.

\subsection{Our Contribution: iMLCA}

To address the difficulty of reporting exact values, we propose \emph{iMLCA}, a new ML-powered combinatorial auction \textit{with interval bidding}. iMLCA asks bidders to report upper and lower bounds on their values (i.e., intervals).  For our design, we take the MLCA mechanism by \citet{brero2019machine} as the starting point but we extensively modify and extend it to  handle the new interval queries. While interval bidding simplifies the interaction with the auction, we must take great care in order to retain high efficiency. Accordingly, a significant aspect of our contribution is to combine an ML-based approach with interval-based bidding into a single coherent design.

The main challenge is that bidders' reported intervals will often be overlapping, such that one cannot determine the efficient allocation nor reasonable payments.  Consequently, iMLCA must force bidders to successively \textit{tighten} their reported intervals, but without requiring bounds to be tightened so much that the benefits of interval bidding are obviated. To this end, we design a \emph{price-based activity rule}, extending earlier work on refinement processes by \citet{Lubin.2008}, which makes for an easily-understood bidder interaction. Using prices in the activity rule implies that quoting the \emph{right} prices is of utmost importance, such that bound tightening only happens where needed. Accordingly, one of our technical contributions is a new algorithm for generating approximate linear clearing prices.\footnote{Note that we only use linear prices to drive the bound refinement process, while the  elicitation (i.e., deciding which queries to ask) is driven by an ML algorithm. This is in contrast to the CCA, where linear prices are used to drive the elicitation process (see our discussion in Section~\ref{sec:discussion} for details).} We carefully integrate our new price-based activity rule with the ML-powered query module designed by \citet{brero2019machine}, such that they can be executed in the same step in each auction round (phase 2 of iMLCA; see Section~\ref{sec:MLAndPreferenceRefinement}).

To ensure that the bounds are eventually tight enough to determine the efficient allocation (given reports), we design an additional \emph{\ConvergencePhase} (phase 3 of iMLCA; see Section~\ref{sec:ConvergencePhase}). In this phase, iMLCA asks bidders to tighten their bounds on a few bundles, until a convergence criterion is met. By setting this  criterion, the auctioneer can select a trade-off between mitigating bidders' incentives to manipulate and bidders' elicitation efforts. Ultimately, the final allocation and prices are determined at the lower bounds only (Section~\ref{subsec:outcomeDetermination}), as bidders are guaranteed to be willing to pay their lower bounds.

In Section~\ref{sec:theoretical analysis}, we study the theoretical properties of iMLCA. Given our design, it is straightforward to prove that iMLCA satisfies \emph{individual rationality} and \emph{no-deficit}. In contrast, the \emph{incentive analysis} requires more care. While iMLCA (like the CCA and MLCA) is not strategyproof, we explain how our design choices for iMLCA lead to good incentives in practice. 

In Section~\ref{sec:experiments}, we provide a detailed experimental evaluation of iMLCA using the spectrum auction test suite (SATS) \citep{weiss2017sats}.  We show that iMLCA achieves the same efficiency as MLCA, but with much lower elicitation cost. Furthermore, we show that iMLCA outperforms the CCA in terms of efficiency in a realistically-sized CA domain. Regarding incentives, we provide experimental evidence showing that bidders cannot benefit by misreporting their bounds. Finally, in  Section~\ref{sec:discussion}, we discuss limitations of our approach as well as potential extensions.

\section{Related Work}
Our work lies in the broad area of research on preference elicitation algorithms (see \citet{sandholm2006preference} for a survey). \citet{lahaie2004applying} and \citet{blum2004preference} were the first to use ML to improve preference elicitation. They showed that for certain classes of valuations, ML can be used to design tractable elicitation algorithms.
Later work by \citet{brero2018bayesian} and \citet{brero2019fast} follows a design paradigm that is similar to ours (i.e., integrating an ML algorithm into the auction itself). They design a Bayesian price-based mechanism, where the main goal is to improve the speed of an iterative CA until it converges to clearing prices. Similarly, \citet{Shen2019LearningToClear} also used ML to learn clearing prices, but in an ad auction context instead of an iterative CA setting. In contrast to these works, we do not aim to directly learn clearing prices (which may not exist in our setting); instead, we use ML to learn the bidders' value functions, and we use this to steer the iterative query process.

Recently, multiple variations of MLCA have been proposed. \citet{weissteiner2020deep} and \citet{Weissteiner2021FourierCAs} demonstrated that allocative efficiency can be further increased by replacing support vector regression (SVR) with neural networks (NNs) and by using value inference via Fourier Analysis, respectively. \citet{weissteiner2022monotone} and \citet{weissteiner2023bayesian} expanded on the use of neural networks by designing ad hoc architectures and modeling uncertainty over values for unelicited bundles to enhance exploration. \citet{estermann2023deep} proposed query strategies to improve the elicitation of the bundle bids on which the machine learning model is first trained. \citet{soumalias2024machine} and \citet{soumalias2022machine} extended the MLCA query module to include \textit{demand queries}, which are commonly used in iterative auctions and \textit{pairwise comparisons}. These comparisons are particularly useful in domains like course allocation, where universities assign course schedules to students. Our work contributes to this line of research by extending the core MLCA design to ask bidders to report bounds instead of exact values on bundles, simplifying the bidding process.

There is also related work using ML for \textit{automated mechanism design} \citep{dutting2015payment,dutting2019optimal,golowich2018deep}. This work has focused on learning allocation and payment rules, so that the resulting mechanism achieves high efficiency or revenue and is (approximately) truthful. Recently, \citet{Brero2021RLForIndirectMechanisms} used reinforcement learning to learn optimal \emph{indirect} mechanisms within the restricted class of sequential price mechanisms. All of these works use ML to learn the mechanism; but in contrast to iMLCA, the final (learned) mechanism does not use ML when executed.

\citet{Lubin2018ComputationalSearch} studied the automated search for good core-selecting mechanisms.
The authors incorporated an equilibrium solver into their search algorithm, such that they could algorithmically decide how close to strategyproof different mechanisms are. Unfortunately, their approach does not yet scale to large settings nor to iterative auction design problems. \citet{Tang2017ReinforcementMD} proposed a similar approach for optimizing mechanisms in dynamic environments like ad auctions. Instead of employing an equilibrium solver, they use ML to model agents' behavior inside the optimization.

\section{Preliminaries}

In this section, we present our formal model and provide a summary of the MLCA mechanism.

\subsection{Formal Model}
Combinatorial auctions (CAs) allocate $m$ indivisible items among $n$ bidders. We let $M = \{1,...,m\}$ denote the set of items and $N = \{1,...,n\}$ the set of bidders. Bundles of items are represented by an indicator vector $x \in \mathcal{X} = \{0,1\}^m$, i.e., $x_j = 1$ iff item $j \in M$ is contained in bundle $x$. Each bidder $i$'s preferences are captured by a value function denoted $v_i: \mathcal{X} \rightarrow \mathbb{R}_{\geq 0}$. We let $v = (v_1,...,v_n)$ denote the \emph{value profile} of all bidders. Without loss of generality, we assume that the value functions are normalized such that the value of the empty bundle is zero; importantly, we impose no further structural restrictions. 

We also refer to the set of all bidders $N$ as the \emph{main economy}. For this economy, an allocation is denoted by $a = (a_1,...,a_n) \in \mathcal{X}^n$, where $a_i$ is the bundle allocated to bidder $i$. We denote the set of \textit{feasible} allocations by $\mathcal{F}=\left\{a \in \mathcal{X}^n:\sum_{i \in N}a_{ij} \le 1, \,\,\forall j \in M\right\}$.
Sometimes, we need to work with a \emph{marginal} economy, where a single bidder $i$ is omitted.  In this case, we denote the allocation in the marginal economy by $a^{-i}=(a_1^{-i},\dots,a_{i-1}^{-i},a_{i+1}^{-i},\dots,a_n^{-i})$, where the superscript ``$-i$'' indicates which bidder has been excluded.

A bidder $i$ may make non-truthful value reports $\hat{v}_i$ to the mechanism, giving rise to value profiles other than the true one. For any profile $\hat{v}$, we let $\hat{v}(a)=\sum_i \hat{v}_i(a_i)$ denote the total reported value of allocation $a$. When the true value profile is used, this quantity is called the \emph{social welfare} of the allocation. An allocation that maximizes the social welfare is denoted $a^* \in \argmax_{a \in \mathcal{F}} v(a)$. We let $p_i$ denote the payment charged to bidder $i$, and we let $p = (p_1,...,p_n)$ denote the payment profile for all bidders. We assume quasi-linear utilities of the form $u_i(a,p) = v_i(a_i) - p_i$. 

For iMLCA, we introduce an \emph{interval query} as a generalization of a \emph{value query}, but where reports specify upper and lower bounds rather than exact values.\footnote{Our interval query is similar to that in the \textit{tree-based bidding language} \citep{cavallo2005tbbl} but simpler, in that bidders need to specify bounds on only a single bundle rather than a concisely represented set of bundles.} We denote the $k$\textsuperscript{th} bundle-value report of bidder $i$ by $(x_{ik},\underline{v}_{ik},\overline{v}_{ik})$, where $\underline{v}_{ik}$ is the lower bound and $\overline{v}_{ik}$ the upper bound reported for bundle $x_{ik}$. Note that we do not use the standard $\hat{v}$ notation here to simplify notation and because intervals are always reports so there is no opportunity for confusion.  We denote the set of all \textit{bundle-value reports} of bidder $i$ as $R_i$; the profile of all reports is $R = (R_1,...,R_n)$. To simplify notation, we say that $x \in R_i$ if there exists $k : (x_{ik},\underline{v}_{ik},\overline{v}_{ik})\in R_i$ and $x_{ik} = x$. %
For convenience and by slightly overloading notation, we denote by $\underline{v}_i(\cdot)$ and $\overline{v}_i(\cdot)$ the \emph{upper and lower reported value functions} for each bidder $i$ respectively; note that these functions have support on $R_i$ only and are elsewhere undefined. We denote the set of feasible allocations based on reports $R$ as $\mathcal{F}_R = \{a\in\mathcal{F} : a_i \in R_i \; \forall i\}$.

We will also need the following valuation function in the course of defining our mechanism:

\begin{defn}[\textsc{Perturbed Valuation} \citep{Lubin.2008}]
Bidder $i$'s \emph{perturbed valuation function} with respect to a given allocation $a$, defined over bundles $x \in R_i$, is given by:
\begin{equation}
\perturbed{v}_i(x|a) = \begin{cases}
\underline{v}_i(x) & \text{if} \; x = a_i \\
\overline{v}_i(x) & \text{otherwise}
\end{cases}
\end{equation}
\end{defn}
\noindent In words, the perturbed valuation enables us to capture the worst case loss from allocating bundle $a_i$ to bidder $i$ instead of any other bundle $x$.  We write $\perturbed{v}_i(x|a)$ as $\perturbed{v}_i(x)$ where $a$ is clear from context.

We use linear prices in our refinement processes, which we denote by $\pi \in \mathbb{R}^m_{\geq 0}$. \textit{Clearing prices} are prices such that demand meets supply, yielding a feasible allocation $a$ with $v_i(a_i) - \pi(a_i) \geq v_i(x) - \pi(x) \; \forall i\in N, x \in \mathcal{X}$ (demand) and $\pi(a) \geq \pi(a') \; \forall a' \in \mathcal{F}$ (supply). 
Together, the clearing prices and such an allocation form a \emph{competitive equilibrium}.


\subsection{MLCA}

\begin{figure}[tb]
\centering
\begin{minipage}[t]{.5\columnwidth}
\centering
\includegraphics[width=1\columnwidth]{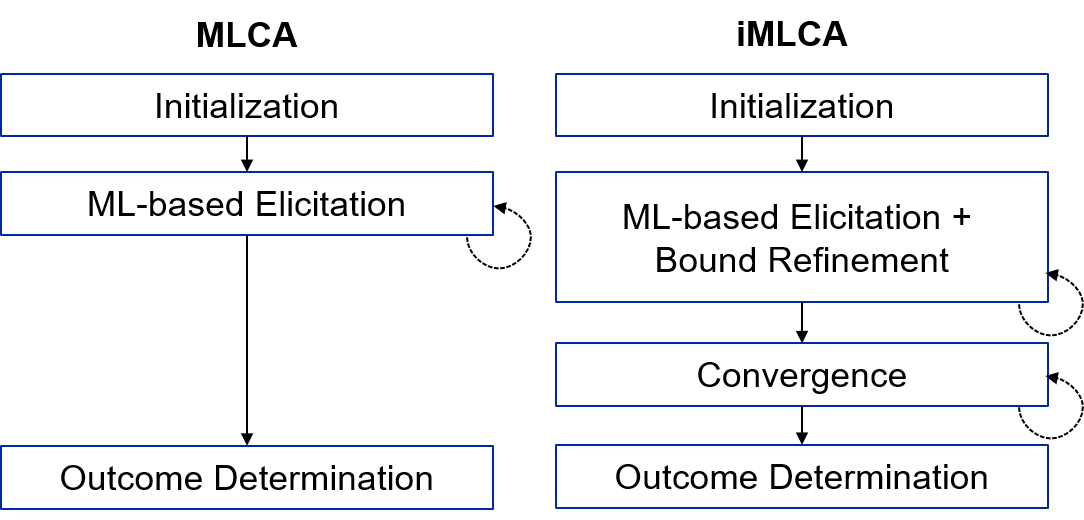}
\caption{Auction phases of MLCA and iMLCA}
\label{fig:flowchart}
\end{minipage}
\end{figure}

The \emph{ML-powered ICA (MLCA)} by \citet{brero2019machine} has the three phases shown in Figure~\ref{fig:flowchart}: 
\begin{enumerate}[itemsep=0px]
\item The \emph{initialization phase}, where random bundles are queried generating input for the ML. 
\item The iterative \emph{ML-based elicitation phase}, where an ML-powered query module is used to generate new queries every round. 
\item The \emph{outcome determination phase}, where the outcome (final allocation and payments) is computed based solely on reported bundle-value pairs.
\end{enumerate}
The core of MLCA is the \textit{ML-powered query module} used in phase 2, which, in every round of the auction, aims to generate a query profile that forms a feasible allocation; such bundles will likely be useful when computing the final allocation. Concretely, the query module uses an ML algorithm to \emph{generalize} from the bundle-value pairs that have already been reported to predict values for all unreported bundles. It then computes the efficient allocation at these learned valuations and queries bidders for their value at this allocation. If the ML algorithm is sufficiently accurate, then the queried allocation will be highly efficient; otherwise, the bidder can correct the ML via the issued query. This way, MLCA proceeds towards a more efficient outcome, assuming a sufficiently effective ML algorithm and truthful bidding. For details, we refer the reader to \cite{brero2019machine}.

\section{ML-powered Combinatorial Auction with Interval Bidding}
\label{sec:intervalquery}

\let\oldnl\nl
\newcommand{\nonl}{\renewcommand{\nl}{\let\nl\oldnl}}

\begin{algorithm}[tb]
\caption{iMLCA}
\DontPrintSemicolon
\label{alg:imlca}
\SetAlgoLined
\Parameter{
max queries $Q\textsuperscript{max}$; 
initial queries $Q\textsuperscript{init}$; 
final allocation interval size limit $\varepsilon\textsuperscript{stop}$;}
\nonl\hrulefill\\
\lForEach(\Comment*[f]{\textcolor{blue}{Initialization Phase}}){$i \in I$\label{alg_line:init_start}}{query $Q\textsuperscript{init}$ random bundles}\label{alg_line:init_end}
\nonl\hrulefill\\
\While(\Comment*[f]{\textcolor{blue}{ML + Bound Refinement Phase}}){$|R_1| \leq Q\textsuperscript{max}$\label{alg_line:ml_start}}{
\ForEach{$i \in I$}{
	Generate  main economy queries for efficiency
	\label{alg_line:mainEcon}
	\Comment*[r]{Alg. 3, \citealt{brero2019machine}}
    Generate marginal econ. queries for incentives \& revenue 
    \label{alg_line:marginalEcon}
	\Comment*[r]{Alg. 3, \citealt{brero2019machine}}
    Query generated bundles and apply \emph{\MRPAR} activity rule 
    \label{alg_line:mrpar}
    \Comment*[r]{Activity Rule~\ref{act:mrpar} using}
    \Comment*[f]{Procedures~\ref{proc:provisonal_allocation} and~\ref{proc:price_new}}
    }
}\label{alg_line:ml_end}
\nonl\hrulefill\\
$\varepsilon = \varepsilon\textsuperscript{stop}$ 
\Comment*[r]{\textcolor{blue}{\ConvergencePhase}}
\While(\Comment*[f]{Procedure~\ref{proc:convergence_achieved}})
{$\neg$ \textup{Convergence Stopping Rule}\label{alg_line:ref_start}}{
	\ForEach{$i \in I$}
	{
	Generate  queries in main economy for convergence  \Comment*[r]{Algorithm~\ref{alg:convergence_bundles}}
	Query these bundles, apply \emph{\IntervalReduction} activity rule using $\varepsilon$ \Comment*[r]{Activity Rule~\ref{act:convergence}}
	}
	$\varepsilon = \varepsilon/2$;\\
}\label{alg_line:ref_end}
\nonl\hrulefill\\
$\underline{a} \in \argmax_{a\in\mathcal{F}_R}\sum_{i\in N}\underline{v}_i\left(a_i\right)$ \label{alg_line:out_start}
\Comment*[r]{\textcolor{blue}{Outcome Determination}}
\lForEach{$i\in N$}{$p_i =  \sum_{j\neq i} \underline{v}_j\left(\underline{a}_j^{-i}\right) - \sum_{j\neq i} \underline{v}_j\left(\underline{a}_j\right)$,
where 
$a\textsuperscript{-i} \in \argmax_{a\in\mathcal{F}_R}\sum_{j\neq i}\underline{v}_j\left(a_j\right)$
\label{alg_line:out_end}}
\KwRet{$\underline{a}$, $p$}\;
\end{algorithm}

In this section, we introduce our new ML-powered ICA with interval bidding (iMLCA).  From a design perspective, iMLCA builds upon the MLCA design and shares some common components; but to facilitate interval bidding, significant changes are necessary.  Figure \ref{fig:flowchart} depicts the differences between the mechanisms.  Specifically, iMLCA enhances the ML phase by adding a bound refinement component that requires bidders to tighten relevant bounds as the auction progresses. Further, iMLCA includes an entirely new third phase that forces further narrowing of the bounds to guarantee convergence (i.e., ensuring that bounds are tight enough such that the efficient allocation at reports can be determined). The overall iMLCA mechanism is provided in Algorithm~\ref{alg:imlca}. Before describing it in detail in the following subsections, we provide a high-level overview:
\begin{enumerate}[itemsep=0px]
\item The \emph{initialization phase} (Line \ref{alg_line:init_start}; described in Section~\ref{subsec:init}) queries random bundles to generate preliminary input for the ML algorithm.

\item The \emph{ML-based elicitation and bound refinement phase} (Lines \ref{alg_line:ml_start}-\ref{alg_line:ml_end}; described in Section~\ref{sec:MLAndPreferenceRefinement}) serves as the primary iterative elicitation phase of the mechanism. Similar to MLCA, the auctioneer iteratively presents new bundles to bidders using an ML-powered query module. However, the interval bids employed in iMLCA do not typically reveal to the auctioneer sufficient information to determine an efficient allocation based on the reported bundles. To address this, we additionally ask bidders to refine their bounds using the \MRPAR activity rule. This rule targets refinement based on approximate market clearing prices, thus focusing the bidders' effort in revealing their preferences on bundles that are likely to clear the market (i.e., be part of an efficient allocation) according to the current reported bounds.

\item The \emph{convergence phase} (Lines \ref{alg_line:ref_start}-\ref{alg_line:ref_end}; described in Section~\ref{sec:ConvergencePhase}) asks bidders to further refine their bounds guided by the \IntervalReduction activity rule. 
The rule guarantees that the auctioneer can determine an efficient final allocation at the reports, which MRPAR is unable to do on its own because its prices may only be approximately clearing.  This phase is typically only responsible for a modest amount of refinement, but is important in providing the theoretical guarantees of the mechanism. 

\item The \emph{outcome determination phase} (Lines~\ref{alg_line:out_start}-\ref{alg_line:out_end}; described in Section~\ref{subsec:outcomeDetermination}) computes the final allocation and payments.
\end{enumerate}

We now describe each of the four phases in detail.

\subsection{Initialization}
\label{subsec:init}

In this phase, each bidder is queried $Q\textsuperscript{init}$ (e.g., 50) bundles selected uniformly at random from the complete bundle space. In contrast to MLCA, the bidders answer these queries using upper and lower bounds. As in MLCA, the resulting reports are used as initial training data for the ML algorithm. We do not impose activity rules in the initialization phase; thus, bidders could initially report wide intervals, but at the cost of being forced to refine more dramatically in later rounds.  

\subsection{Machine Learning-based Elicitation and Bound Refinement}
\label{sec:MLAndPreferenceRefinement}

In this phase, the bidders are iteratively asked in rounds to both provide new bounds on additional bundles, and to refine their bounds on previously reported bundles.  
While notionally simple, there are several important details in realizing this design. Before diving into the details, we provide an overview of the four steps we will go through.

First, we will describe how queries for new bundles are generated, using an ML-based technique from \citet{brero2019machine}, but modified to use interval queries (Section~\ref{subsec:genNewQueries}).  Next, we will describe the \MRPAR rule which is used to force bidders to narrow their bounds in each round (Section~\ref{subsec:boundRefinement}).  The \MRPAR rule requires bidders to respond to a \emph{provisional allocation}, the determination of which we describe next (Section~\ref{subsec:provAlloc}). 
Finally, \MRPAR also uses approximate clearing prices at reports in forcing revelation, essential for its economic motivation; we introduce a new method for selecting such prices that are better targeted at driving effective elicitation (Section~\ref{subsec:generatingLinearPrices}).

\subsubsection{Generating new Queries}
\label{subsec:genNewQueries}

iMLCA seeks to add bundles for both the \emph{main} and the \emph{marginal} economies (Lines \ref{alg_line:mainEcon} and \ref{alg_line:marginalEcon} of Algorithm~\ref{alg:imlca} respectively). 
In both economies, iMLCA requires bidders to respond to queries in each round, and it uses ML to inform which bundles to query.  
The actual selection of bundles via ML is achieved via the \emph{ML-powered query module} of MLCA \citep[Algorithm 3]{brero2019machine}, but in iMLCA, these queries are \emph{interval queries} instead of the standard value queries used in MLCA. Switching for interval queries means that the ML algorithm used in the query module must generalize from \emph{bounds} on each data point in its training data, rather than exact values. iMLCA works with any ML algorithm capable of doing this.  

We adopt support vector regressions (SVRs) as our ML algorithm, as they can be modified to use bounds instead of exact values in their training data.  Specifically, the standard \emph{$\varepsilon$-insensitive hinge loss} function used in SVRs only penalizes the learner for predictions more than $\varepsilon$ away from the exact training data.  It is mathematically straightforward to replace this $2\varepsilon$-wide insensitivity region from the standard formulation with a data point-specific region capturing the bounds (see \citet{brero2018combinatorial}).  Thus, in our usage, the SVR is only penalized for predictions outside of $[\underline{v},\overline{v}]$.   

\subsubsection{Bound Refinement}
\label{subsec:boundRefinement}

In every round, iMLCA not only asks bidders about \textit{new} bundles but also asks them to \textit{tighten} their bounds on already-queried bundles $R_i$, according to an \textit{activity rule} that balances elicitation against bidder effort (Algorithm~\ref{alg:imlca}, Line~\ref{alg_line:mrpar}).  Tighter bounds serve two purposes: first, they aid the ML algorithm in making better predictions, which improves the selection of new bundles to query. Second, given the price-based method used to drive them, they facilitate progress towards finding the most efficient allocation among those bundles that have been queried. We next provide details on this latter aspect and the activity rule that enables it.

Under the \emph{Modified Revealed Preference Activity Rule (\MRPAR)} \citep{Lubin.2008} bidders are provided both their portion of a \emph{provisional allocation} $a_i^\alpha$ (Procedure~\ref{proc:provisonal_allocation} below), and a set of prices $\pi$ (Procedure~\ref{proc:price_new} below).  As the allocation is provided by the auctioneer, its feasibility can be enforced. Bidders must then make clear that the provisional bundle is their weakly preferred bundle at the quoted prices, or alternatively that some other bundle is strictly preferred.  
Formally, we have:
\begin{activityrule}%
\begin{samepage}
\vspace{-.75ex}
{\textsc{Modified Revealed Preference Activity Rule (\MRPAR) \citep{Lubin.2008}:}}\label{act:mrpar} \\
Given a provisional allocation $a^\alpha$, linear prices $\pi$, and existing reports $R_i$, each bidder $i$ must submit reports revealing their preferred bundle $x^\ast_i$ by satisfying the following:
\begin{equation}
\exists x^{\ast}_i \in R_i \ \text{s.t.}\ 
\underline{v}_i(x^{\ast}_i) - \pi(x^{\ast}_i) \geq  \overline{v}_i(x_i) - \pi(x_i) \forall x_i \in R_i \setminus \{x^{\ast}_i\}  
\label{eq:reveal}
\end{equation}
\noindent where the inequality in Equation~\eqref{eq:reveal} must be strict if $x^{\ast}_i \neq a^\alpha_i$.
\label{eq:mrpar-second}
\vspace{-.75ex}
\end{samepage}
\end{activityrule}

\refMRPARActivityRule sits within a class of activity rules that employ the principle of \emph{revealed preference} (see \citet{ausubel2017practical} for details).  In general, such rules work by forcing bidders to reveal their preferred bundle at quoted prices, and then the mechanism updates prices in a direction that helps supply meet demand.  Through this movement, the  mechanism will eventually obtain clearing prices and a corresponding allocation that form a competitive equilibrium. When this happens, a special case of the First Welfare Theorem \citep[16.C-D]{mas1995microeconomic} provides that the allocation will be efficient.
Exact linear clearing prices may not exist in a combinatorial setting \citep{bikhchandani2002package}; thus, iMLCA may have only approximate clearing prices to work with.  Fortunately, \citet[][Lemma 1]{Lubin.2008} have anticipated this and shown that under approximate clearing prices we still obtain approximate efficiency.%
\footnote{The lemma is proven in a context where bidders state their full valuation function using a compact bidding language.  By contrast, in our setting, the quoted efficiency is only with respect to the bundles under consideration, $R$.}
This highlights the importance of finding good approximately clearing prices $\pi$
(Section~\ref{subsec:generatingLinearPrices}, below).

While satisfying Equation~\ref{eq:reveal} may seem complicated, Figure~\ref{fig:mrpar} illustrates that the intuition behind the rule, moving bounds to identify the most preferred bundle at given prices, is easily visualized. Practical implementations of iMLCA would include bidding interfaces that provide such a visualization as an aid to bidders.  We emphasize that, when responding to \refMRPARActivityRule, bidders must only consider the small subset of bundles under consideration $R_i$, not the set of all possible bundles.

\begin{figure}[t]
\centering
\includegraphics[width=0.45\columnwidth]{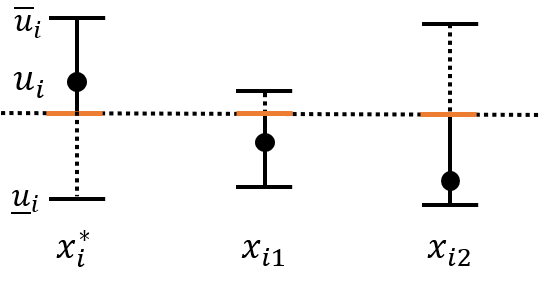}
\caption{
\label{fig:mrpar}
\small
Illustration of \refMRPARActivityRule. We see the scenario before refinement in black. The upper bars are the upper bound utility $\overline{u}_i(x) = \overline{v}_i(x)-\pi(x)$, the lower bars are the lower bound utility $\underline{u} = \underline{v}_i(x) - \pi(x)$ and the dots are the true utility $u_i(x)$ induced by prices $\pi$. A bidder must refine their bounds such that the lower bound utility on their preferred bundle $x^{\ast}$ is at the same level as the upper bound utility on all their other bundles (illustrated in orange).
}
\end{figure}

\subsubsection{Generating a Provisional Allocation} 
\label{subsec:provAlloc}
To instantiate \refMRPARActivityRule, we must first select a
\emph{provisional allocation}.  Following \citet{Lubin.2008}, we choose the allocation by identifying a value profile that is an affine combination of the existing upper and lower bounds. 

\begin{algoprocedure}%
\begin{samepage}
\textsc{Generate Provisional Allocation \citep{Lubin.2008}}\label{proc:provisonal_allocation}\ \\
For a chosen $\alpha \in [0,1]$, we compute a provisional valuation 
\begin{equation}
v_i^\alpha(x) = \alpha\underline{v}_i(x)+(1-\alpha)\overline{v}_i(x) \quad \forall i \in N, x \in R_i
\end{equation}
and a provisional allocation 
\begin{equation}
a^\alpha \in \argmax_{a\in\mathcal{F}} \sum_{i\in N}v_i^\alpha(a_i).
\end{equation}
\end{samepage}
\end{algoprocedure}
As in \citet{Lubin.2008}, we set $\alpha = \max\{0.5,\omega\}$, where $\omega$ is a convergence bound we will introduce in Section~\ref{sec:ConvergencePhase}.  As iMLCA proceeds, the provisional allocation moves towards the lower bounds (which are the ones ultimately implemented).

\subsubsection{Generating Linear Prices}
\label{subsec:generatingLinearPrices}
\refMRPARActivityRule requires linear prices.  Ideally, these would be \emph{clearing prices} with respect to the reported values.  However, in a CA, linear clearing prices may not exist.  We therefore seek \textit{$\delta$-approximate linear clearing prices} $\pi$ (at the given reports).  

\begin{defn}[\textsc{$\delta$-approximate Clearing Prices At Reports}]
\label{defn:delta_prices}
Given reported valuation profile $\hat{v}$ and an allocation $a$, prices $\pi$ are \emph{$\delta$-approximate clearing prices at reports} if 
$\delta \geq \delta_{ik} \forall i,k$,
\begin{align}
    &\hat{v}_i(a_i) - \pi(a_i) + \delta_{ik} \geq \hat{v}_i(x_{ik}) - \pi(x_{ik}) \forall i,k, \quad \text{and}
    \label{eq:delta_constraint1} \\
    &\pi_j = 0 \quad \forall j \in M : \nexists i \in N : a_{ij} = 1.
    \label{eq:zero_prices1}
\end{align}
\end{defn}
Equation~\eqref{eq:delta_constraint1} ensures that the clearing condition is violated by no more than $\delta_{ik}$ for any reported bundle $x_{ik}$, and $\delta$ is defined to bound all such $\delta_{ik}$.
Equation~\eqref{eq:zero_prices1} sets the price of any un-allocated item to zero, which ensures that the supply condition for clearing prices is met.
Because it may easily be that  $\delta_{ik} \neq \delta$ for many $k$, many price vectors may meet Definition~\ref{defn:delta_prices}.  We first break ties in favor of minimizing the number of positive $\delta_{ik}$ (see Procedure~\ref{proc:prices} in \appref{app:deltaApproxClearing}).
We then break ties further, by selecting a price vector that reduces the number of bundles a bidder must consider in \refMRPARActivityRule (as we will show formally in Section~\ref{sec:theoretical analysis}).  Prices are thus selected as follows:

\begin{algoprocedure}
\textsc{Generate prices}
\label{proc:price_new}
\begin{enumerate}[label=\roman*.]
\item Generate $\delta$-approximate clearing prices at reports with respect to the provisional valuation $v^\alpha$ and allocation $a^\alpha$; let $\delta_{ik}^\alpha$ be the required approximation (see Procedure~\ref{proc:prices} in \appref{app:deltaApproxClearing}). \label{proc_line:e1} 
\item Minimize over all $\delta_{ik}^\alpha > 0$ according to the objective $||\delta^\alpha_{ik}||_2$, while keeping constraints generated in step~\ref{proc_line:e1}
\item Generate $\delta$-approximate clearing prices at reports with respect to the provisional allocation $a^\alpha$ and the perturbed valuation $\perturbed{v}$ (defined with respect to $a^\alpha$) while keeping all previously generated constraints; let $\perturbed{\delta}_{ik}$ be the required approximation. \label{perturbed_price_step}
\item Minimize $||\perturbed{\delta}_{ik}+C||_2$, holding all constraints active.  Here $C > \max_{i,k}-\delta_{ik}$ is a constant added to ensure that negative $\delta_{ik}$ are properly ordered. \label{proc_line:minimize_perturbed_quadratic}
\item Solve a new program where the objective in step~\ref{proc_line:minimize_perturbed_quadratic} becomes a constraint and we maximize the sum of the prices. 
\end{enumerate}
\end{algoprocedure}

In Section~\ref{subsubsec:effect_of_price_generation}, we confirm the effectiveness of our price generation procedure by comparing it to a simpler benchmark inspired by \citet{kwasnica2005new} and \citet{Lubin.2008}. Our enhanced procedure significantly reduces bidder effort, as measured by the number of bundles considered during bounds refinements, while also improving allocative efficiency.

\subsection{Convergence}
\label{sec:ConvergencePhase}

When the second phase is complete, bidders will have been guided to report and refine bounds on bundles that are important for determining efficiency via \refMRPARActivityRule.  However, we may not have elicited enough information to determine the efficient allocation at reports: \refMRPARActivityRule's guarantee of convergence is approximate in the $\delta$ of the $\delta-$approximate linear clearing prices. We therefore continue into the third phase, where the \IntervalReduction activity rule ensures bidders further refine the bounds on the bundles that have already been reported in $R$ such that iMLCA can perfectly identify the bundles belonging to an efficient outcome at reports.  Our \IntervalReduction activity rule thus serves a similar purpose to the earlier DIAR rule in \citet{Lubin.2008}, but is much simpler to implement and obey.  Moreover, as we shall see in Section~\ref{subsubsec:effect_of_convergence_phase}, it requires less bidder revelation.  The rule is as follows:

\begin{activityrule}
\textsc{\IntervalReduction Activity Rule\label{act:convergence}}\ \\
Given $\varepsilon$ and bundle set $Q_i$, bidder $i$ must submit reports such that $\frac{\overline{v}_i(x)-\underline{v}_i(x)}{\overline{v}_i(x)} \leq \varepsilon \;\forall\; x \in Q_i$.
\end{activityrule}
The rule forces revelation according to a relative measure (i.e. $\frac{\overline{v}(x)-\underline{v}(x)}{\overline{v}(x)}$, the ratio of interval size and the upper bound value), and is parameterized by a set of bundles over which it operates, $Q_i$, and a constant $\varepsilon$, which specifies how much tightening is required.  In each round, $\varepsilon$ is reduced by $50\%$ in Algorithm~\ref{alg:imlca}. Therefore, $\varepsilon \rightarrow 0$ in the number of rounds, at which point full revelation over the bundles in $Q_i$ would occur.
However, Algorithm~\ref{alg:imlca} does not proceed to force revelation all the way to this limit.  Instead it may stop much sooner, and with significant slack in the bounds in $Q_i$. The stopping condition we employ for the \ConvergencePhase is as follows:

\begin{algoprocedure}
\begin{samepage}
\textsc{Convergence Stopping Rule}\label{proc:convergence_achieved}\\[1ex]
$\omega=\frac{\underline{v}(\underline{a})}{ \perturbed{v}(\perturbed{a})} = 1$ 
and
$\frac{\overline{v}(\underline{a})-\underline{v}(\underline{a})}{\overline{v}(\underline{a})} \leq \varepsilon\textsuperscript{stop}$     
\end{samepage}
\end{algoprocedure}

This stopping rule has two requirements.  First, it ensures that the bounds have been narrowed enough to calculate the efficient allocation at reports.  Specifically, after each round, it computes the convergence bound from \citet[Equation 26]{Lubin.2008}: $\omega = \underline{v}(\underline{a})/ \perturbed{v}(\perturbed{a}) \leq 1$.  Recall that the perturbed valuation~$\perturbed{v}$ can be used to measure the maximal surplus of choosing any bundle $x_i$ over $a_i$.  Accordingly, the allocation with maximal such surplus is $\perturbed{a} \in \argmax_{a\in \mathcal{F}_R}{\perturbed{v}(a|\underline{a})}$.  Thus, when $\omega=1$, convergence is achieved and the efficient allocation at reports can be guaranteed. 
Second, as we will discuss in Section~\ref{subsec:incentives}, the size of the reported intervals for the final allocation pertains to the mechanism incentives.  It is thus desirable for the auctioneer to have control over the permitted gap size, $\overline{v}(\underline{a}) - \underline{v}(\underline{a})$, at the end of elicitation.  To provide this control, the \ConvergencePhase stops only when the relative reported interval $\frac{\overline{v}(\underline{a}) - \underline{v}(\underline{a})}{\overline{v}(\underline{a})}$ is no larger than a chosen parameter $\varepsilon\textsuperscript{stop}$.

\begin{algorithm}[tb]
\caption{Identify Relevant Bundles for Convergence for Bidder $i$}
\DontPrintSemicolon
\label{alg:convergence_bundles}
\SetAlgoLined
\Parameter{ 
bidder $i$;
interval size limit $\varepsilon$; Queries per round $Q\textsuperscript{round}$
}
$Q_i := \{\}$ \Comment*[r]{The set of bundles to query} 
$E_i := \{\}$ \Comment*[r]{The set of excluded bundles} 
\lIf*{$\frac{\overline{v}_i(\underline{a}_i)-\underline{v}_i(\underline{a}_i)}{\overline{v}(\underline{a}_i)} > \varepsilon$}{
    $Q_i := Q_i \cup \{\underline{a}_i\}$ \label{alg_line:convergenceTightCheck1}
    \Comment*[r]{Query $\underline{a}_i$ if interval is not yet tight enough}
    } 
    \lElse* {
    $E_i := E_i \cup \{\underline{a}_i\}$ \Comment*[r]{Otherwise exclude from being queried}
}
\While{$|Q_i| < Q\textsuperscript{round}$\label{alg_line:convergenceLoopStart}}{
    $\mathcal{F}' := \{a\in \mathcal{F}_R : a_i \not\in\ Q_i \cup E_i\}$; \label{alg_line:convergenceFprime} \\
    $\perturbed{a}' \in \argmax_{a\in \mathcal{F}'}\sum_i{\perturbed{v}_i(a_i)}$;
    \label{alg_line:convergenceFprimeAlloc}\\
\lIf*{$\frac{\overline{v}_i(\perturbed{a}_i')-\underline{v}_i(\perturbed{a}_i')}{\overline{v}(\perturbed{a}_i')} > \varepsilon$}{
    $Q_i := Q_i \cup \{\perturbed{a}_i'\}$ 
    \label{alg_line:convergenceTightCheck2}
    \Comment*[r]{Query $\perturbed{a}_i'$ if interval is not yet tight enough}
    } 
    \lElse* {
    $E_i := E_i \cup \{\perturbed{a}_i'\}$ \Comment*[r]{Otherwise exclude from being queried}
}
}\label{alg_line:convergenceLoopEnd}
\KwRet{$Q_i$}
\end{algorithm}

It remains to describe the set of bundles $Q_i$ over which the \IntervalReduction activity rule operates, which we identify via Algorithm~\ref{alg:convergence_bundles}.  The algorithm is parameterized to ensure bidders are provided with $Q^\textsuperscript{round}=|Q_i|$ bundles.
Given the structure of $\omega$, tighter bounds on either $\underline{a}$ or $\perturbed{a}$ will achieve progress towards convergence. Accordingly, the algorithm considers including these allocations in Lines~\ref{alg_line:convergenceTightCheck1} and~\ref{alg_line:convergenceTightCheck2}, if their reported interval is not yet tight enough.  The loop in Lines~\ref{alg_line:convergenceLoopStart}-\ref{alg_line:convergenceLoopEnd} selects not only the most preferred bundle $\perturbed{a}_i$ at $\perturbed{v}_i$, but also the next-best and next-next-best, etc. (Lines~\ref{alg_line:convergenceFprime} and~\ref{alg_line:convergenceFprimeAlloc}), until $|Q_i|=Q^\textsuperscript{round}$. 

\subsection{Outcome Determination}
\label{subsec:outcomeDetermination}

iMLCA determines the final allocation based on the set of reported bundles only; thus, the set of feasible allocations is $\mathcal{F}_R = \{a\in\mathcal{F} : a_i \in R_i \; \forall i\}$. The final allocation $\underline{a}$ is determined by  $\argmax_{a\in\mathcal{F}_R} \sum_{i\in N}\underline{v}_i(a_i)$.  When determining the allocation, we take a conservative approach by using the lower bound valuation $\underline{v}$ (Line~\ref{alg_line:out_start} of Algorithm~\ref{alg:imlca}). This guarantees that we never implement an outcome that is higher than a bidder's \textit{true} value (under truthful bidding), which ensures individual rationality. Finally, we also use $\underline{v}$ to compute VCG-style payments (Line~\ref{alg_line:out_end} of Algorithm~\ref{alg:imlca}). We discuss the incentive implications of this design in Section~\ref{subsec:incentives}. 

\subsection{Additional Design Features of iMLCA}
\label{sec:AdditionalDesignFeatures}
We equip iMLCA with several additional design features that may be important in some domains. First, similarly to \citet{brero2019machine}, we allow bidders to provide additional information to the mechanism by reporting unsolicited bundle-value pairs in the initialization phase; we refer to these as ``push bids.''  Second, we allow the auctioneer to select alternative payment rules to VCG, e.g., the \textit{VCG-nearest} rule \citep{day2012quadratic}, which we also test in our experiments.

\section{Theoretical Analysis}
\label{sec:theoretical analysis}

In this section, we present a theoretical analysis of iMLCA. 
We begin by analyzing standard mechanism design properties, specifically individual rationality (Subsection~\ref{subsec:IR}), no-deficit (Subsection~\ref{subsec:no_deficit}), and bidders' incentives (Subsection~\ref{subsec:incentives}). Next we analyze the implication of our new price-based activity rule on bidder effort (Subsection~\ref{sec:BidderEffortReduction}). Finally, we discuss bounding the efficiency loss in the performance of the ML algorithm (Section~\ref{sec:effLoss}).

\subsection{Individual Rationality} \label{subsec:IR}

A mechanism satisfies individual rationality if no bidder pays more than their reported value.

\begin{proposition}[Individual Rationality]
iMLCA satisfies individual rationality.
\end{proposition}

\begin{proof}
The reported lower bound utility of each bidder $i$ is non-negative as $$\underline{v}_i\left(\underline{a}\right) - p_i = \underline{v}\left(\underline{a}\right) - \sum_{j\neq i}\underline{v}_j\left(a_j^{-i}\right) \geq 0.$$
Indeed, the marginal allocation $a_j^{-i}$ is also a feasible allocation for the main economy (i.e $a_j^{-i} \in \mathcal{F}_R$) but the allocation with the highest lower bound value in the main economy is $\underline{a} \in \argmax_{a\in \mathcal{F}_R}\underline{v}\left(a\right)$. Therefore we have 
$\underline{v}\left(\underline{a}\right) \geq \sum_{j\neq i}\underline{v}_j\left(a_j^{-i}\right).$
\end{proof}

\subsection{No-Deficit}
\label{subsec:no_deficit}

A mechanism satisfies no-deficit if the sum of all payments is weakly positive. We have:

\begin{proposition}[No-Deficit]
iMLCA satisfies no deficit.
\end{proposition}

\begin{proof}
The payment of each bidder $i$ is non-negative, i.e., 
$p_i =  \sum_{j\neq i} \left( \underline{v}_j\left(a_j^{-i}\right) - \underline{v}_j(\underline{a}_j)\right) \geq 0.$
This is the case because $a_j^{-i} \in \mathcal{F}_R$ and $a_j^{-i} \in \mathcal{F}_R$ are both feasible allocations with respect to all reports $R$, but $a_j^{-i} \in \argmax_{a\in \mathcal{F}_R}\sum_{j\neq i}\underline{v}_i(a)$ is the allocation that maximizes the lowerbound value for the marginal economy. Thus, we have 
$ \sum_{j\neq i}\underline{v}_j\left(a_j^{-i}\right) \geq \sum_{j\neq i}\underline{v}_j(\underline{a}_j).$
\end{proof}

\subsection{Incentives}
\label{subsec:incentives}
In this section, we discuss the incentives of iMLCA. Like all spectrum auction designs used in practice \citep[e.g., CCA;][]{ausubel2017practical}, iMLCA is not strategyproof. This is not surprising, given that MLCA (using exact value reports) does not provide such a guarantee \citep{brero2019machine}. However, we can offer a parallel argument to that of \citet{brero2019machine} to argue that, in practice, iMLCA provides bidders with good incentives to report truthfully by aligning their utilities with the overall economic welfare. We start by defining a \textit{truthful strategy} for bidders reporting upper and lower bounds.
\begin{defn}[Truthful Strategy in iMLCA]
    In iMLCA, a bidder’s strategy is called \textit{truthful} if the bidder only reports upper and lower bounds that are consistent with their true values.\footnote{Note this includes both bounds for bundles queried by the auction or ``pushed'' by the bidder (see Section~\ref{sec:AdditionalDesignFeatures}).} 
\end{defn}

 Because iMLCA charges VCG payments, the utility of each bidder $i$ can be calculated as follows:
\begin{equation}\label{eq:utility_imlca}
u_i = \underbrace{v(\underline{a}_i) + \sum_{j\neq i} \underline{v}_j\left(\underline{a}_j\right)}_{\begin{subarray}{c}
\text{(I) welfare main economy}
\end{subarray}
} - \underbrace{\sum_{j\neq i} \underline{v}_j\left(a_j^{-i}\right)}_{\begin{subarray}{c}\text{(II) welfare marginal economy}\end{subarray}}
\end{equation}
Term (I) is the sum of $i$'s true value and the reported lower bound of the other bidders for $\underline{a}$, and term (II) is the reported lower bound of the other bidders for $a^{-i}$. If bidder $i$ wants to increase their utility, they must increase the difference (I) - (II). 

Next, we distinguish two pathways for manipulations in iMLCA:
\begin{enumerate}
    \item Bidder $i$ tries to change the allocation in their main or marginal economy.
    \item Bidder $i$ tries to affect the tightness of other  bidders' bounds in their respective main or marginal economies.
\end{enumerate}
The first pathway already existed in MLCA. The second pathway is new in iMLCA. We now explain why iMLCA is robust against both manipulation pathways in practice.

\noindent\textbf{Changing the allocation in the main or marginal economy.} A bidder may try to decrease term (II) by affecting which allocation the mechanism eventually uses for their marginal economy. However, following the same argument as for MLCA, we note that (II) is practically independent of bidder $i$'s report because iMLCA explicitly generates queries for each bidder's marginal economy (i.e., excluding all previous reports by bidder $i$). Given this, the only remaining way for bidder $i$ to improve their utility via an allocation change is to bias the final allocation $\underline{a}$ towards an allocation that has higher social welfare. But as in MLCA, iMLCA allows bidders to submit ``push bids'' in the initialization phase, enabling them to push information to the auction that they deem useful. Thus, if they have a bundle in mind which they believe will increase social welfare, they can simply submit a corresponding push bid. For more details, please see \citet{brero2019machine}.

\noindent\textbf{Affecting the tightness of other bidders' bounds.} In iMLCA, even without affecting an allocation change, a bidder can try to increase their utility by (a) inducing other bidders to tighten their bounds as much as possible on the final allocation $\underline{a}$ (i.e., increasing term (I)), or by (b) inducing other bidders to keep their bounds in their marginal economy as loose as possible (i.e., reducing term (II)). iMLCA explicitly addresses (a) as follows: In the \ConvergencePhase, the auctioneer can control the size of the reported intervals for the final allocation directly by choosing the parameter $\varepsilon\textsuperscript{stop}$. If $\varepsilon\textsuperscript{stop}$ is small enough, a bidder cannot force tighter bounds and thus iMLCA prevents such manipulations by design.%
\footnote{An auctioneer can balance information revelation (i.e., tighter intervals) and incentive properties by choosing $\varepsilon\textsuperscript{stop}$. Moreover, $\varepsilon\textsuperscript{stop}$ only affects a small subset of reports and the intervals of all other reports can remain loose.} 
Regarding (b), we note that iMLCA only requests bound refinements for bundles relevant for the final allocation $\underline{a}$. Thus, as there is no explicit bound refinement process in the marginals, any remaining manipulation strategies seem implausible to execute in practice.  We offer experimental evidence to this effect in Section~\ref{subsec:exp_bidder_incentives}.

By making the following assumptions, we can derive more formal incentive guarantees:
\begin{assumption}\label{assumption:independent_marginal}
	For every bidder $i$, if all other bidders report truthfully, then the reported welfare of bidder $i$’s marginal economy is independent of bidder $i$’s value reports.
\end{assumption}
As we have discussed above, explicitly querying each bidder $i$'s marginal economy makes $a^{-i}$ practically independent of their reports, supporting (but not guaranteeing) Assumption~\ref{assumption:independent_marginal}.
\begin{assumption}\label{assumption:lower_bounds}
	When all bidders are truthful, all bidders' lower bounds for the final allocation reported at the end of the auction, are equal to the true values, i.e., for each bidder $i$: $\underline{v}_i(a_i) = v_i(a_i)$.
\end{assumption}
By setting a sufficiently low $\varepsilon\textsuperscript{stop}$, the auctioneer can guarantee that, for each bidder $i$, the equation $\underline{v}_i(a_i) = v_i(a_i)$ gets arbitrarily close to holding with equality, supporting Assumption~\ref{assumption:lower_bounds}.
Assumptions~\ref{assumption:independent_marginal} and \ref{assumption:lower_bounds} allow us to extend the social welfare alignment property of MLCA to iMLCA:
\begin{proposition}[Social Welfare Alignment]
\label{Prop:SocialWelfareAlignment}
If Assumptions~\ref{assumption:independent_marginal} and \ref{assumption:lower_bounds} hold and  all other bidders are truthful, then a bidder can only increase their utility by increasing the reported social welfare of $\underline a$.
\end{proposition}
\begin{proof}
    The proof follows in a straightforward way by considering Equation~\eqref{eq:utility_imlca}. Because of Assumption~\ref{assumption:independent_marginal}, we know that each bidder cannot affect the reported welfare in their marginal economy. From Assumption~\ref{assumption:lower_bounds}, and because bidders are truthful, we have that the term (I) in Equation~\eqref{eq:utility_imlca} corresponds to the social welfare in the allocation $\underline{a}$. Thus, the only way for a bidder to increase their utility is by 
    increasing the reported social welfare for the final allocation $\underline a$.
\end{proof}
When combined with the opportunity to submit push bids, social welfare alignment already provides strong incentives for truthful reporting.
Concretely, instead of trying to game the mechanism to lead it to a more efficient allocation, a bidder can push truthful information they has about promising bundles to the auctioneer.
%
%
That said, social welfare alignment is not sufficient to guarantee that a bidder cannot \textit{accidentally} drive the mechanism to a better allocation via manipulations, without having this allocation in mind. To address this last scenario and obtain a formal equilibrium property, we need one more assumption:
\begin{assumption}\label{assumption:efficient_main}
    When all bidders are truthful then the final allocation of iMLCA is efficient.
\end{assumption}
As efficiency can be bounded in the quality of the ML algorithm (see Section~\ref{sec:effLoss} below), the assumption relies on the quality of the ML-algorithm.  Using SVRs, Assumption~\ref{assumption:efficient_main} holds in the majority of auction instances for two of the domains we study (see Section~\ref{sec:experimental_results}), but not for the third one. If Assumption~\ref{assumption:efficient_main} indeed holds, then we obtain the following result.
\begin{proposition}[Truthful Reporting is an Ex-post Nash Equilibrium]
    If Assumptions~\ref{assumption:independent_marginal}, \ref{assumption:lower_bounds}, and \ref{assumption:efficient_main} hold, truthful reporting is an ex post Nash equilibrium of the game induced by iMLCA.
\end{proposition}
\begin{proof}
    The proof follows by considering that, under Assumptions~\ref{assumption:independent_marginal} and \ref{assumption:lower_bounds}, iMLCA is social welfare aligned. 
    Using Assumption~\ref{assumption:efficient_main}, we have that this social welfare (and thus bidder $i$'s utility) is maximized when bidder $i$ is truthful. This concludes the proof.
\end{proof}

\subsection{Bidder Effort Reduction}
\label{sec:BidderEffortReduction}

\begin{figure}
\centering
\includegraphics[width=0.33\columnwidth]{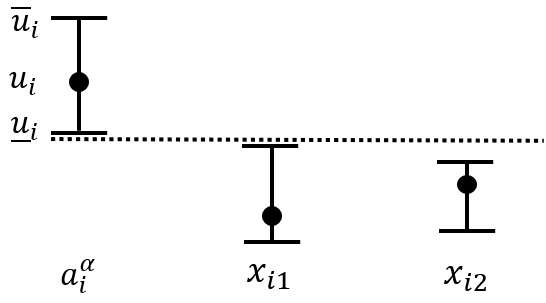}
\caption{%
\label{fig:mrpar_effort_reduction}
\small
Illustration of bundles that need not be considered by a bidder $i$ when responding to \refMRPARActivityRule. Shown is the scenario prior to refinement for given prices $\pi$. The bidder need not consider $x_{i1}$ and $x_{i2}$ as  it is immediate from the location of the bounds that neither can be their preferred bundle.}
\end{figure}
When a bidder is asked to refine their reports according to \refMRPARActivityRule, there may be many bundles that need not be considered at all. Specifically, bidder $i$ need not consider any bundle $x_i$ for which the utility interval induced by the given prices $\pi$ does not overlap with the utility interval of the provisionally allocated bundle. This effect is illustrated in Figure~\ref{fig:mrpar_effort_reduction}. The following proposition shows that all bundles for which $\perturbed{\delta}_{ik} \leq 0$ need not be considered according to \refMRPARActivityRule. 

\begin{proposition}[Bidder Effort] Bidder $i$ can ignore all bundles $x_{ik} \in R_i~\setminus~\{a_i^\alpha\}$ with $\perturbed{\delta}_{ik} \leq 0$ when responding to \refMRPARActivityRule.
\label{lemma:bidder_effort}
\end{proposition}
\begin{proof}
Before refinement, it holds that 
$
\underline{v}_i(a_i^\alpha) - \pi(a_i^\alpha) \geq \overline{v}_i(x_{ik}) - \pi(x_{ik})
$
for all bundles $x_{ik} \neq a_i^\alpha$ with $\perturbed{\delta}_{ik} \leq 0$. 
But according to \refMRPARActivityRule bidder $i$ must report some $x^{\ast}_i \neq x_{ik}$ such that   
$
\underline{v}_i(x^{\ast}_i) - \pi(x^{\ast}_i) \geq \underline{v}(a_i^\alpha) - \pi(a_i^\alpha) \geq \overline{v}_i(x_{ik}) - \pi(x_{ik}).
$
\end{proof}

As we maximize the number of such $\perturbed{\delta}_{ik}$ in Step~\ref{perturbed_price_step} of Procedure~\ref{proc:price_new}, we are guaranteed to reduce the bidder effort in terms of bundles to consider during an \refMRPARActivityRule refinement, when compared to alternative price generation procedures that omit this step \citep[e.g., that of][]{kwasnica2005new}.  Our new linear price selection algorithm is helping us achieve our goal of reducing bidder elicitation effort; it may also be of independent use in other revealed preference-based mechanisms.  

\subsection{Bounding the Efficiency Loss}
\label{sec:effLoss}

We can bound the efficiency loss of the mechanism in the prediction performance of the ML-algorithm, following a similar result for MLCA \citep{brero2018combinatorial}.  Surprisingly, even though iMLCA only has interval information to work with, it can reason about those intervals sufficiently well that the  bound on efficiency is no weaker than that for MLCA, as provided in the following proposition:

\begin{proposition}[Efficiency Bound]\label{prop:learning_error_learned}
Assume that bidders' interval reports are consistent with their true values.  Let $\tilde{v}$ be some valuation profile learned by the ML algorithm, let $\tilde a$ be an efficient allocation w.r.t. to $\tilde{v}$, and let $a^*$ be an efficient allocation w.r.t. the true valuation profile.  Assume that the learning errors in $\tilde{a}$ and $a^*$ are bounded as follows: for each bidder $i$, $|\tilde{v}_i(\tilde{a}) - v_i(\tilde{a})|\le\delta_1$ and $|\tilde{v}_i(a^*) - v_i(a^*)|\le\delta_2$, for $\delta_1,\delta_2\in\mathbb R$. Then the following bound on the efficiency loss in the final allocation $\underline{a}$ holds for all $\tilde{v}$ learned for any economy, in any iteration of iMLCA:

\begin{equation}\label{eq:learning_error}
1 - \frac{v(\underline{a})}{v(a^*)} \le  \frac{n(\delta_1+\delta_2)}{v(a^*)}
\end{equation}
\end{proposition}
\begin{proof}
Let $a^+ \in \max_{a\in F_R} v(a)$ be the allocation that maximized the welfare with respect to reports and assuming truthful reports we have according to \citet{brero2019machine} 
\begin{equation}
\frac{v(a^+)}{v(a^*)} \geq 1-\frac{n(\delta_1+\delta_2)}{v(a^*)}
\end{equation}
and according to \citet{Lubin.2008} 
\begin{equation}
\frac{v(\underline{a})}{v(a^+)} \geq \frac{\underline{v}(\underline{a})}{\perturbed{v}(\perturbed{a})}.
\end{equation} Multiplying these inequalities and from the fact that $\frac{\underline{v}(\underline{a})}{\perturbed{v}(\perturbed{a})} = 1$ at the end of iMLCA due the stopping condition of the \ConvergencePhase (Procedure \ref{proc:convergence_achieved}) the proposition follows.
\end{proof}

\section{Experimental Evaluation}
\label{sec:experiments}

In this section, we evaluate the performance of iMLCA by comparing it against the original MLCA and the CCA.  We consider MLCA to be an upper bound benchmark for iMLCA, given that MLCA is similar in design but works with more precise information. We also conduct ablation studies to assess the impact of iMLCA's design features, including its price generation procedures and \ConvergencePhase.

\subsection{Experiment Set-Up}

\subsubsection{Domains.} For our experiments, we use the spectrum auction test suite (SATS) version 0.8.0.
\citep{weiss2017sats}. SATS enables us to generate thousands of synthetic auction instances on demand, providing access to bidders' values, and allowing us to quickly compute the true efficient allocation. We use the following three domains:

\begin{description}
\item[Global Synergy Value Model (GSVM)]\citep{goeree2010hierarchical}: A relatively simple medium-sized domain with 18 items and 7 bidders of different types. 

\item[Local Synergy Value Model (LSVM)]\citep{scheffel2012impact}: A medium-sized domain with 18 items and 6 bidders of different types. Synergies arise from the spatial proximity of items, which makes LSVM more complex.

\item[Multi-Region Value Model (MRVM)]\citep{weiss2017sats}:  A large-sized domain with 98 items and 10 bidders of different types. It models the complex 700Mhz Canadian spectrum auction.
\end{description}

\subsubsection{Modeling Bidder Behavior.}
When simulating bidders' responses to queries in our experiments, we model bidders as behaving consistently with their true preferences. Thus, for every bidder $i$ and every bundle $x \in R_i$, it holds that: $\underline{v}_i(x) \leq v_i(x) \leq \overline{v}_i(x)$. As we do not want the reported intervals $[\underline{v}_i(x) , \overline{v}_i(x)]$ to always be centered around the true value $v_i(x)$, we simulate the interval reports following the same approach as \citet{brero2018combinatorial}: For any given bundle $x$ and bidder $i$, we independently draw two error measures $z_1$ and $z_2$ from a standard normal distribution with mean 0 and standard deviation $\mu$, where $\mu$ is a parameter capturing the bidder's reported interval size. The lower and upper bounds are then obtained as $\underline{v}_i(x) = \max(0,v_i(x)(1-z_1))$ and $\overline{v}_i(x) = v_i(x)(1+z_2)$, respectively. In our experiments, we use $\mu = 0.5$. To provide some intuition: with $\mu = 0.5$, a typical interval report is $[100, 200]$. To simulate the process by which bidders tighten their bounds, we have implemented two straightforward heuristics that adjust bounds to meet \refMRPARActivityRule and \refIntervalReductionActivityRule respectively. For details, please see \appref{app:bidderRefinement}.

\subsubsection{Mechanism Set-Up.}
For iMLCA, we use the same parameters as were used by \citet{brero2019machine} for MLCA\footnote{The parameters adopted from \citet{brero2019machine} include the SVR parametrization, $Q\textsuperscript{init}=50$, $Q\textsuperscript{max}=500$ ($Q\textsuperscript{max}=100$ for GSVM) and the number of interval queries per round and bidder in the \emph{ML-based elicitation and bound refinement phase} $Q\textsuperscript{round}$ is set to 4. The same $Q\textsuperscript{round}$ is used as parameter for Algorithm \ref{alg:convergence_bundles}. }.
We set the final interval size limit $\varepsilon\textsuperscript{stop}=0.005$. For example, this corresponds to a maximal interval of the reported social welfare of $[99.5, 100]$. The refinement parameter $\alpha$ is obtained as described in \citet{Lubin.2008}. To solve all optimization problems, we use CPLEX (version 12.10) and adopt a time limit for some problems (including the new price generation problem) similar to \citet{brero2019machine}.\footnote{We used an Ubuntu 16.04 cluster with AMD EPYC 7702 2.0 GHz processors with 8 cores and 32 GB of RAM. We use the best sub-optimal solution found when the time limit is reached.} Our source code is publicly available under an open-source license at \url{https://github.com/marketdesignresearch/mechlib}.

\subsection{Results: Comparing iMLCA, MLCA and CCA}
\label{sec:experimental_results}

Table~\ref{tbl:iMLCA} provides a comparison of iMLCA, MLCA and CCA in all three domains.\footnote{To enable a head-to-head comparison, we use the same setting seeds (for the test set) as used by \citet{brero2019machine} which are seeds $101-200$ for GSVM and LSVM, and seeds $101-150$ for MRVM.} The efficiency of an allocation $a$ is measured as $v(a)/v(a^*)$ and revenue as the total payment relative to the maximal social welfare $\sum_i p_i / v(a^*)$. We further report alternative VCG-\emph{nearest} core selecting payments \citep{day2012quadratic} and the total number of rounds (including all phases of iMLCA).

\begin{table}[tb]
\centering
\begin{tabular}{ccrrrrr}
Domain & Mechanism & Efficiency & Revenue & Revenue (Core) & Rounds \\ 
\hline 
\hline 
\multirow{3}{*}{GSVM}
 & iMLCA & 100.0\% (0.01)
 & 58.5\% (1.32)	
 & 63.1\% (1.22)	
 & 16 (0.1)
 \\ 
\cline{2-6}
 & MLCA & 100.0\% (0.00) 
 & 68.4\% (1.10)	
 & 72.4\% (1.05)	
 & 14 (0.0) \\ 
\cline{2-6}
 & CCA & 100.0\% (0.00) 
 & 68.1\% (1.13)
 & 73.1\% (0.99)
 & 235 (3.0)
 \\ 
\hline 
\hline
\multirow{3}{*}{LSVM}
 & iMLCA & 99.6\% (0.10)
 & 74.4\% (1.22)
 & 78.7\% (1.12)
 & 122 (1.2)
 \\ 
\cline{2-6}
 & MLCA & 99.6\% (0.12) 
 & 80.9\% (0.91)	
 & 84.5\% (0.83)	
 & 114 (0.0) \\ 
\cline{2-6}
 & CCA & 99.9\% (0.03)
 & 82.3\% (0.91)
 & 86.4\% (0.73)
 & 125 (0.3) 
 \\ 
\hline 
\hline
\multirow{3}{*}{MRVM}
 & iMLCA & 96.2\% (0.15)
 & 23.8\% (0.66)	
 & 25.1\% (0.57)
 & 122 (0.4)
 \\ 
\cline{2-6}
 & MLCA & 96.4\% (0.13)
 & 40.5\% (0.33)	
 & 40.7\% (0.34)
 & 114 (0.0) \\ 
\cline{2-6}
 & CCA & 94.2\% (0.20)
 & 30.0\% (0.77)
 & 34.5\% (0.39)
 & 142 (0.9)
 \\ 
\end{tabular} 
\caption{Results for iMLCA, MLCA and CCA with a maximal number of reports of 500 (100 for GSVM). Standard errors in parentheses. Averages over 100 instances (50 for MRVM).}
\label{tbl:iMLCA}
\end{table}

We focus our analysis on efficiency, as spectrum auctions are generally government-run auctions that target efficiency and not revenue maximization \citep{cramton2013spectrum}.
We use a one-way ANOVA to test for statistically significant differences between all three mechanisms, and we use a post-hoc Tukey test for further pairwise comparisons.

In GSVM, all three mechanisms perform extremely well with 100.0\% efficiency (with no  statistically significant difference; $p>0.1$). Next, we consider LSVM, which is more complex and thus more difficult to learn with the kernel-based SVRs we employ \citep[see][]{brero2019machine}. Here, we find a statistically significant difference between the three mechanisms ($p<0.05$). Further analysis shows that the CCA performs better than MLCA ($p<0.05$). However, CCA and iMLCA do not perform statistically significantly different ($p=0.08$). Furthermore, MLCA and iMLCA do not perform statistically significantly different ($p=0.84$)  

To put the results into perspective, we note that GSVM and LSVM are very simple test domains. In fact, the regional bidders in GSVM have value for less than 64 bundles. Thus, the CCA can elicit those bidders' full value functions, because they can report their value for up to 100 bundles in the supplementary round. This gives the CCA a natural advantage in these simple domains compared to MLCA and iMLCA, which do not have information about the bidders' bundles of interest and must therefore contend with the full bundle space. In light of this, it is particularly impressive that iMLCA does not perform worse than the CCA in GSVM nor in LSVM.

Next, we consider MRVM, where we find a statistically significant difference between the three mechanisms ($p<0.001$). When we compare the mechanisms pairwise, we find no statistically significant difference between the efficiency of MLCA and iMLCA ($p=0.65$). By contrast, the efficiency of the CCA is lower than the efficiency of iMLCA ($p=0.001$) and of MLCA ($p=0.001$).

\begin{remark}
Our experimental results also show that iMLCA achieves slightly lower revenue than the other mechanisms.  This is due to our activity rules driving more elicitation in the main rather than the marginal economies. We discuss potential ways to increase iMLCA's revenue in Section~\ref{sec:discussion}.
\end{remark}

Summarizing our main experiments, we conclude that iMLCA achieves the same efficiency as MLCA, even though it operates with interval queries instead of exact queries. Furthermore, iMLCA achieves a 2\% point higher efficiency than the CCA in the realistically-sized MRVM domain.

\subsection{Reducing Bidder Effort}
\label{subsec:reducing_bidder_effort}

In this section, we show that iMLCA achieves a reduction in bidder effort over MLCA, which was our main design goal and which is iMLCA's main advantage.

\subsubsection{Evaluating Reported Interval Size} 
\label{subsubsec:interval_size}

\begin{table}[tb]
\begin{center}
\begin{tabular}[t]{cccc}
Mechanism & Domain &  \makecell{Initial\\Interval\\Size} & \makecell{Final\\Interval\\Size}\\ 
\hline
\hline 
\multirow{3}{*}{iMLCA}
& GSVM &  55\% & 46\% \\ 
\cline{2-4}
& LSVM & 55\% & 48\% \\ 
\cline{2-4}
& MRVM & 55\% & 41\% \\ 
\hline 
MLCA & all domains & 0\% & 0 \% \\ 
\end{tabular}
\caption{Comparison of average reported relative interval size between iMLCA and MLCA after the \emph{initialization phase} and at the end of the auction. Mean of 100 instances (50 for MRVM).}
\label{tbl:interval_size}
\end{center}
\end{table}

As our measure for bidder effort, we introduce the relative size of all reported intervals. Concretely, for a report $(x_{ik},\underline{v}_{ik},\overline{v}_{ik})$ we measure the \textit{relative interval size} as $(\overline{v}_{ik}-\underline{v}_{ik})/\overline{v}_{ik}$.
In Table~\ref{tbl:interval_size}, we compare MLCA and iMLCA in terms of the initial and final interval size (averaged over the reported bundles) for all bundles with a positive upper bound value (i.e., all bundles a bidder is interested in). For MLCA, the interval is always $0$, because bidders must report exact values. In contrast, for iMLCA, the final interval size is between 46\% and 41\%. Note that even 41\% corresponds to remarkably loose intervals. To see this, consider that with an average final interval size of $41\%$, a typical interval is $[59,100]$. In \appref{app:intervalSizeDistribution}, we also report the distribution over these interval sizes. These results highlight that iMLCA allows bidders to maintain large intervals on the majority of their reported values despite iteratively asking bidders to refine their bounds and forcing relatively tight intervals for the final allocation.
%

\subsubsection{Effect of New \ConvergencePhase}
\label{subsubsec:effect_of_convergence_phase}

Next, we evaluate the effectiveness of our novel \ConvergencePhase. To do this, we compare it against a version of iMLCA where the \ConvergencePhase is replaced by the DIAR activity rule as proposed by \citep{Lubin.2008}. We refer to this version as iMLCA-DIAR. iMLCA-DIAR stops the refinement when $\frac{\underline{v}(\underline{a})}{ \perturbed{v}(\perturbed{a})} = 1$, i.e. when the allocation that is guaranteed to maximize efficiency at reported values has been identified. The results are shown in Table \ref{tbl:iMLCA-diar}. We observe that iMLCA-DIAR needs more than a hundred additional rounds of elicitation in every domain (and several hundred in MRVM) to achieve convergence when compared to iMLCA with the new \ConvergencePhase. Furthermore, final interval sizes for iMLCA-DIAR are tighter than in iMLCA with the new \ConvergencePhase, i.e. bidders are required to reveal far more information. 

\begin{table}[tb]
\begin{center}
\begin{tabular}{ccrr}
\centering
Domain & Mechanism & Rounds & Final Interval Size\\ 
\hline 
\hline 
\multirow{2}{*}{GSVM}
 & iMLCA & 16 (0.1) & 46\% \\ 
\cline{2-4}
 & iMLCA-DIAR & 135 (20.4) & 34\% \\ 
\hline 
\hline
\multirow{2}{*}{LSVM}
 & iMLCA & 122 (1.2) & 48\% \\ 
\cline{2-4}
 & iMLCA-DIAR & 317 (29.9) & 39\% \\ 
\hline 
\hline
\multirow{2}{*}{MRVM}
 & iMLCA & 122 (0.4) & 41\% \\ 
\cline{2-4}
 & iMLCA-DIAR & 496 (59.5) & 33\% \\ 
\hline 
\end{tabular} 
\caption{Results for iMLCA and iMLCA-DIAR with a maximal number of reports of 500 (100 for GSVM). Standard errors in parentheses. Averages over 100 instances (50 for MRVM).}
\label{tbl:iMLCA-diar}
\end{center}
\end{table}

\subsubsection{Evaluation of Enhanced Price Generation Procedure}
\label{subsubsec:effect_of_price_generation}

A key difference between iMLCA and MLCA is the requirement for bidders to refine their bounds based on auction prices during each round of the machine learning-based elicitation phase. To take advantage of this structure, it is important for the auction to quote prices that reduce the number of bundles bidders need to consider in their refinements.

In this section we assess the impact of our new price generation procedure (Procedure~\ref{proc:price_new}) by comparing it to a simpler alternative inspired by \citet{kwasnica2005new}. This simpler procedure uses $\delta$-approximate clearing prices with respect to $v^\alpha$ only.\footnote{The price generation procedure proposed by \citet{Lubin.2008} is also based on \citet{kwasnica2005new}, while they apply different tie-breaking rules to approximate final payments.} 
We refer to the auction using this baseline procedure as iMLCA-K and evaluate iMLCA and iMLCA-K in terms of the average number of bundles bidders consider per \refMRPARActivityRule refinement and the final allocation efficiency, our primary objective.

Table~\ref{tbl:effort} summarizes our findings. Despite mixed trends in the smaller domains, our new price generation objective significantly reduces the number of bundles bidders need to consider during \refMRPARActivityRule refinements in the realistically sized MRVM domain. Additionally, in MRVM, iMLCA demonstrates significantly higher efficiency than iMLCA-K ($p<0.01$). This contrast is not observed in the smaller domains where differences in efficiency are not statistically different ($p=0.23$ in GSVM and $p=0.62$ in LSVM).

\begin{table}[tb]
\centering
\begin{tabular}{lccc}
Domain & Mechanism & Efficiency &
 {\centering \# Refinements} \\
\hline
\hline
\multirow{2}{*}{GSVM} & iMLCA & 100\% (0.01) & 6.6 (0.11) \\ 
\cline{2-4}
 & iMLCA-K & 100\% (0.00) & 7.6 (0.09) \\
\hline
\hline
\multirow{2}{*}{LSVM} & iMLCA & 99.6\% (0.10) & 3.5 (0.11) \\
\cline{2-4}
 & iMLCA-K & 99.7\% (0.09) & 3.2 (0.13) \\
\hline 
\hline
\multirow{2}{*}{MRVM} & iMLCA & 96.2\% (0.15) & 26.4 (0.51) \\
\cline{2-4}
 & iMLCA-K & 95.5\% (0.16) & 35.4 (0.29) \\
\hline 
\end{tabular} 
\caption{Effect of iMLCAs' new price generation compared to a version of iMLCA using a price generation procedure inspired by \citet{kwasnica2005new} (iMLCA-K). Standard errors in parentheses. Averages over 100 instances (50 for MRVM).}
\label{tbl:effort}
\end{table}

\subsection{Bidder Incentives}
\label{subsec:exp_bidder_incentives}

Next, we provide experimental evidence about the incentive properties of our mechanism.

\subsubsection{\IntervalReduction Activity Rule}
\label{subsec:exp_interval_reduction_activity_rule}

\begin{table}[b]
\begin{center}
\begin{tabular}[t]{crr}
Domain & \makecell{True\\Manipulability\\Bound} & \makecell{Observable\\Manipulability\\Bound} \\ 
\hline 
\hline 
GSVM & 0.1\% (0.01)
 & 0.2\% (0.01)		 \\ 
\hline 
LSVM & 0.0\% (0.00)
 & 0.0\% (0.01) \\ 
\hline 
MRVM & 0.0\% (0.00)	
 & 0.0\% (0.01)	 \\ 
\end{tabular}
\caption{
The \emph{true manipulability bound} reports the relative difference of the lower bound and true value.  
The \emph{observable manipulability bound} reports the relative difference of the lower and upper bound.  Both measures total over the final allocation. Standard errors in parentheses. Mean of 100 instances (50 for MRVM).}
\label{tbl:eff_gap}
\end{center}
\end{table}

As discussed in Section \ref{subsec:incentives}, to prevent manipulations in the main economy, it is important that the reported lower bounds for the final allocation are close to the true values. While the auctioneer has no knowledge of the size of this interval, she can control the observed interval size (i.e., the difference between upper and lower bounds) via the parameter $\varepsilon\textsuperscript{stop}$. 
We report both intervals in Table~\ref{tbl:eff_gap}. The relative \textit{true manipulability bound} reports the tightness of the lower bound with respect to the true value, i.e. $v(\underline{a})-\underline{v}(\underline{a}))/v(\underline{a})$, while the relative \textit{observable manipulability bound} is  $\overline{v}(\underline{a})-\underline{v}(\underline{a}))/\overline{v}(\underline{a})$.
We see that, in all domains, the true manipulability bound is $\leq 0.1\%$, which means that the mechanism leaves no practically meaningful opportunities for manipulating the bounds of the final allocation.

\subsubsection{Bound Manipulation Experiments}
\label{subsubsec:exp_manipulation}

Given the interval-based nature of iMLCA, bidders now have a new potential manipulation path that was not present in MLCA: the ability to strategically manipulate their bounds. To support our incentive analysis (see Section~\ref{subsec:incentives}), we ran experiments to explore this potential manipulation. In these experiments, we hold all bidders' bound-reporting strategies fixed (as described in Section~\ref{sec:experiments}), other than a single deviating bidder. We then evaluate each of the following bound generation heuristics for this deviating bidder: \emph{Exact Values}, where the focal bidder reports exact values; \emph{Standard Bounds}, where the focal bidder uses the standard bound-reporting strategy where the random noise in the initial bounds is created with a standard deviation of $0.5$; \emph{Wider Bounds}, where the deviating bidder instead uses a standard deviation of $1$; and \emph{Maximum Bounds}, where the deviating bidder reports our numerical equivalent to $[0,\infty]$.

\begin{table}[tb]
\centering
\resizebox{\textwidth}{!}{%
\begin{tabular}{p{3cm}||p{0.8cm}p{1.7cm}p{1cm}|p{.8cm}p{1.7cm}p{1cm}|p{.8cm}p{1.7cm}p{1cm}}
 & \multicolumn{3}{c|}{Local} & \multicolumn{3}{c|}{Regional} & \multicolumn{3}{c}{National} \\
\hhline{~---------}
{} & Social Welfare & Marginal Economy Welfare
(at Lower Bounds) &      Utility & Social Welfare & Marginal Economy Welfare
(at Lower Bounds) &       Utility & Social Welfare & Marginal Economy Welfare
(at Lower Bounds) &           Utility \\
\hhline{=::=========}
Standard Bounds (SD=0.5) &     9842 (129) &                      9841 (129) &  0.09 (0.08) &     9842 (129) &                      9822 (128) &  18.63 (5.27) &     9842 (129) &                      7468 (152) &  2373 (110) \\
\hhline{-||---------}
Exact Values (SD=0)      &     9814 (126) &                      9813 (126) &  0.04 (0.03) &     9830 (124) &                      9806 (123) &  21.91 (6.30) &     9733 (125) &                      7384 (156) &  2345 (116) \\
\hhline{-||---------}
Wider Bounds (SD=1)      &     9825 (131) &                      9824 (131) &  0.08 (0.08) &     9841 (129) &                      9821 (128) &  18.71 (5.60) &     9830 (127) &                      7441 (158) &  2388 (109) \\
\hhline{-||---------}
Maximum Bounds (0, Inf)  &     9814 (127) &                      9813 (127) &  0.00 (0.00) &     9804 (129) &                      9791 (129) &  12.19 (4.00) &     9815 (127) &                      7457 (150) &  2357 (104) \\
\hhline{-||---------}
ANOVA p-value                  &       0.998 &                        0.998 &     0.691 &       0.996 &                        0.998 &      0.637 &       0.931 &                         0.981 &          0.993 \\
\end{tabular}
}
\caption{Results for different bound reporting strategies for one bidder of each reported type in MRVM. Values in Millions. Standard errors in parentheses. Averages over 50 instances.}
\label{tbl:iMLCA-stability-mrvm}
\end{table}

We report the results of these experiments for the MRVM domain in Table~\ref{tbl:iMLCA-stability-mrvm}, where we observe that bidders have limited ability to improve their utility by strategically manipulating their reported bounds. Across the various deviation strategies we found negligible differences in social welfare, marginal economic welfare, and bidder utility.
To confirm these findings statistically we conducted one-way ANOVA tests, as reported in the last row of the table. The high p-values obtained for bidder utility (all greater than 0.637) strongly suggest that it's  unlikely for a bidder to benefit significantly from such manipulation in the MRVM domain. Similar results were observed for the GSVM and LSVM domains, which we defer to \appref{app:manipulationResults} along with a more detailed description of the various bidding strategies investigated.

\section{Discussion}
\label{sec:discussion}

Our experimental evidence provides strong support for the efficacy of iMLCA, achieving the same efficiency as MLCA but with lower bidder effort. Specifically, using linear prices to drive refinement over a restricted set of bundles in iMLCA appears to be as effective as the exact value queries in MLCA.

Our experimental results on the manipulability of iMLCA require care in interpretation.  These experiments do not evaluate the efficacy of non-truthful behavior \textit{in equilibrium}, but rather the efficacy of a unilateral deviation from our baseline 
strategy. We believe that these experiments are a valuable second-best, given that a full equilibrium analysis of a practical iterative combinatorial auction (ICA) for large settings is currently out of reach.  

In our bound manipulation experiments, we have only studied bidders strategically reporting narrower or wider bounds, but we have not simultaneously considered letting bidders report intervals that imply a value misreport (e.g., considering the set of manipulation strategies analyzed by \citet{brero2019machine}). Studying the cross-product of value misreporting strategies and bound manipulation strategies is computationally very challenging, but would be interesting future work.  

In our experiments, iMLCA exhibits lower revenue than MLCA and the CCA. This is not surprising, given that our activity rules are designed to drive more elicitation in the main economy than in the marginal economies. Thus, the reported bounds in the marginal economy allocations may remain looser, resulting in smaller VCG payments (Line~\ref{alg_line:out_end} of Algorithm~\ref{alg:imlca}). For the spectrum auctions that motivate this work it is efficiency, not revenue, that is the design goal.
But for other settings, revenue could be important. Thus, in future work, the activity rules could be altered to seek additional refinement in the marginal economies, which would raise revenue.

\section{Conclusion}

In this paper, we have developed a new ML-powered ICA with interval bidding (iMLCA), which avoids the high costs associated with reporting exact values. Our experiments have demonstrated that iMLCA achieves the same efficiency as MLCA, and outperformed the CCA in the realistically-sized MRVM domain. This suggests iMLCA as a practical mechanism to conduct large ICAs. We emphasize that our goal was not to beat MLCA, but rather to offer comparable performance while requiring less information from the bidders (intervals instead of exact values).

On a technical level, our main contribution was the careful integration of a price-based refinement process with an ML-powered query module in one auction phase. To address incentives, we highlight our new \ConvergencePhase that enables the auctioneer to make a deliberate trade-off between elicitation effort and robustness to manipulation.

We have used one specific measure of bidder effort, namely the relative size of the reported intervals. One could also consider other measures (e.g., from the behavioral economics literature), including non-linear measures capturing that narrowing already small bounds may be harder than narrowing loose bounds. One interesting question for future work would be to investigate how iMLCA could be modified to directly target such non-linear bidder effort measures. It would be interesting to complement these efforts with lab experiments, studying how bidders perceive iMLCA compared to other auction designs. 

The recent FCC Incentive Auction \citep{connolly2018fcc}, highlights that there are important settings that generalize combinatorial auctions to combinatorial double auctions and combinatorial exchanges.  The ICE mechanism \citep{Lubin.2008}, which was inspiration for several of iMLCA's design choices, supports exchange environments.  It is thus natural and of real-world interest to explore the potential of generalizing iMLCA to accommodate combinatorial double auctions and exchanges, thereby expanding its applicability to a wider range of settings.

\clearpage
\section*{Appendices}
\appendix
\section{Compute \texorpdfstring{$\delta$}{Delta}-approximate Clearing Prices at reports}
\label{app:deltaApproxClearing}
\section*{} 
\begin{algoprocedure}%
\textsc{Compute $\delta$-approximate Clearing Prices at reports}\label{proc:prices}\\
Solve a sequence of linear programs to find $\pi$.  First:
\begin{align}
\delta = & \min \max_{i,k} \delta_{ik} \label{eq:min_delta} &\\
\mathrm{s.t.} \, &
\hat{v}_i(a_i^\alpha) \!-\! \pi(a_i^\alpha) \!+\! \delta_{ik} \geq \hat{v}_i(x_{ik})\!-\!\pi(x_{ik}) \forall i,k & \label{eq:delta_constraint}\\
&\pi_j = 0 \quad \forall j \in M : \nexists i \in N : a^\alpha_{ij} = 1&
\label{eq:zero_prices}
\end{align}
Where $\delta_{ik}$ measures violation of the clearing condition for 
the $k$th report of bidder $i$.  Thus, the overall maximal negative utility of any bidder given prices $\pi$ is $\delta$.  Next solve:
\begin{align}
\min &|\{\delta_{ik} : \delta_{ik} > 0 \}| & \label{eq:minimize_number_of_positive_deltas} \\
\mathrm{s.t.} \, & \eqref{eq:delta_constraint} \text{ and } \eqref{eq:zero_prices} \\
&\delta_{ik} \leq \delta \quad \forall i,k & \label{eq:fix_maximal_delta}
\end{align}
which breaks ties to minimize the number of positive $\delta_{ik}$.
\end{algoprocedure}

\section{Modeling Bidder Behavior - Refinement}
\label{app:bidderRefinement}

We implemented two straightforward heuristics that adjust bounds according to \refMRPARActivityRule and \refIntervalReductionActivityRule, respectively.  For the former, we run the \emph{\MRPAR heuristic} (Algorithm~\ref{alg:mrpar_heuristic}), for the latter, we run the \emph{\IntervalReduction activity rule heuristic} (Algorithm~\ref{alg:interval_reduction_heuristic}). The \MRPAR heuristic reflects that it is typically harder for a bidder to narrow a bound close to their exact value for a bundle than it is to narrow a bound far from their true value.  In both heuristics, we aim to minimize the amount of refinement, such that bidders only refine the bounds as much as is needed to comply with the activity rules.

\newpage
\begin{algorithm}[H]
\caption{MRPAR Heuristic}
\label{alg:mrpar_heuristic}
\SetAlgoLined
\DontPrintSemicolon
\Parameter{provisional allocation $a^\alpha$; linear prices $\pi$; $R_i$ the reports of bidder $i$}
\Comment*[l]{Identify $\hat{x}$ the bundle with the highest utility}
\Comment*[l]{Break ties in favour of $a_i^\alpha$}
$\hat{x} \in \argmax_{x \in R_i} v_i(x) - \pi(x)$\;
\Comment*[l]{Identify $\hat{\hat{x}}$ the bundle with the second highest utility}
$\hat{\hat{x}} \in \argmax_{x \in R_i \setminus \{\hat{x}\}} v_i(x) - \pi(x)$\;
\Comment*[l]{Draw a break point utility $\hat{u}$}
\Comment*[l]{D(x,y) is a (capped) bell-shaped distribution which draws random from the range [x,y]}
$\hat{u} \sim D(v_i(\hat{\hat{x}}) - \pi(\hat{\hat{x}}), v_i(\hat{x}) - \pi(\hat{x})$\;
\Comment*[l]{Make sure refinement bounds are not tightened more than needed}
\uIf{$\max_{x\in R_i \setminus \{\hat{x}\}}\overline{v}_i(x)-\pi(x) < \hat{u}$}{
$\hat{u} = \max_{x\in R_i \setminus \{\hat{x}\}}\overline{v}_i(x)-\pi(x)$\;
}
\uIf{$\hat{u} < \underline{v}_i(\hat{x}) - \pi(\hat{x})$}{
    $\hat{u} = \underline{v}_i(\hat{x}) - \pi(\hat{x})$
}
\Comment*[l]{Make sure strict inequality is met in case $\hat{x} \neq a_i^\alpha$}
$\eta = 0$\;
\uIf{$\hat{x} \neq a_i^\alpha$ and $\underline{v}_i(\hat{x}) - \pi(\hat{x})\leq \max_{x\in R_i \setminus \{\hat{x}\}}\overline{v}_i(x)-\pi(x)$}{ set $\eta$ to a very small number}
\Comment*[l]{Adjust bounds}
$\underline{v}_i(\hat{x}) = \min \{v_i(\hat{x}), \hat{u} + \pi(\hat{x}) + \eta \} \big\}$\;
\lForEach{$x \in R_i \setminus \{\hat{x}\}$} {
$\overline{v}_i(x) = \max \big\{v_i(x), \min \{\overline{v}_i(x), \hat{u} + \pi(x) - \eta \} \big\}$ 
}
\KwRet{$R_i$}\;
\end{algorithm}

\begin{algorithm}[H]
\caption{\IntervalReduction activity rule heuristic}
\label{alg:interval_reduction_heuristic}
\SetAlgoLined
\DontPrintSemicolon
\Parameter{bundle to refine $x$; parameter $\varepsilon$; $R_i$ the reports of bidder $i$}
\Comment*[l]{Draw $z$ to identify which part of $\varepsilon$ is above the true value}
\Comment*[l]{D(x,y) is a (capped) bell-shaped distribution which draws random from the range [x,y]}
$z \sim D(0,1)$\;
\Comment*[l]{Generate new upper bound such that the report is consistent (i.e. the new upperbound is weakly lower than the old one and according to $z$. If the interval induced by $z$ would be too small (i.e. the lowerbound is already tighter to the true value than induces by $z$) set the upperbound such that the relative interval is $\varepsilon$ according to the existing lowerbound report.}
$\overline{v}_i(x) = \min\{\max\{\frac{v_i(x)}{1-z\varepsilon}, \frac{\underline{v}_i(x)}{1-\varepsilon} \}, \overline{v}_i(x)\}$ \;
\Comment*[l]{Generate new lower bound such that the report is consistent (i.e. the new lower bound is weakly higher than the old one and the (relative) interval size matches given $\varepsilon$. }
$\underline{v}_i(x) = \max\{\overline{v}_i(x)(1-\varepsilon), \underline{v}_i(x)\}$

\KwRet{$R_i$}\;
\end{algorithm}

\section{Distribution of Interval Size}
\label{app:intervalSizeDistribution}

In Figure~\ref{fig:interval_size} we provide a comparison of the interval size distribution between the initial and final reports in the MRVM domain.  We see that while the reported intervals of the final implemented allocation are almost tight, most others remain loose.  Specifically, at the end of the auction only $19\%$ of the reports have an interval size of $0\%$ to $10\%$ while $37\%$ of the reports have a interval size of more than $50\%$.

\captionsetup[sub]{font=footnotesize,labelfont={bf,sf}}
\begin{figure}[H]
\centering
\begin{subfigure}{.42\textwidth}
  \centering
  \includegraphics[width=1\linewidth]{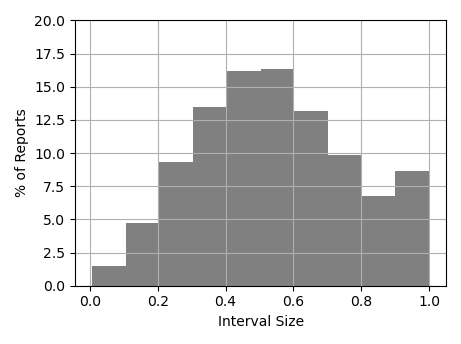}
  \caption[font=footnotesize]{Initial Interval Size}
  \label{fig:initial_uncertainty}
\end{subfigure}%
\begin{subfigure}{.42\textwidth}
  \centering
  \includegraphics[width=1\linewidth]{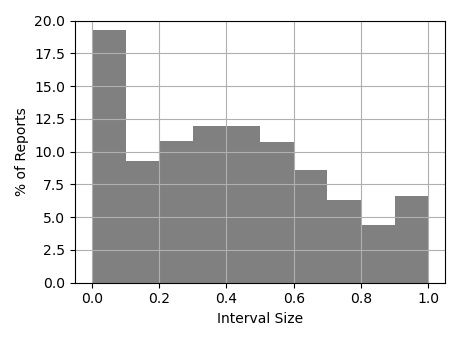}
  \caption{Final Inverval Size}
  \label{fig:final_uncertainty}
\end{subfigure}
\caption{Comparison of the reported interval size distribution after the initialization phase and at the end of the auction.}
\label{fig:interval_size}
\end{figure}

\section{Additional Manipulation Experiments}
\label{app:manipulationResults}

To support our incentive analysis (see section \ref{subsec:incentives}), we ran experiments to explore if a bidder can improve their utility by strategically reporting their bounds.  In these experiments we hold all bidders' bound-reporting approach fixed at our standard procedure (as described in \ref{sec:experiments}), other than a single deviating bidder.  We then evaluate each of the following bound generation heuristics for this deviating bidder:\footnote{All of these heuristics are 'safe' in that they the bidder's true value always lies within their reported bounds}  

\begin{description}
\item[Exact Values (SD=0)] The deviating bidder reports exact values, i.e. for each report for bundle $x$ we have $\underline{v}(x) = \overline{v}(x) = v(x)$. 

\item[Standard Bounds (SD=0.5)] The bidder simply uses the standard approach, as described in Section~\ref{sec:experiments}. Here, SD refers to the standard deviation used when drawing the noise used to generate bounds and we use our baseline level of $SD=0.5$.

\item[Wider Bounds (SD=1)] The deviating bidder generates bounds where the interval size is in expectation doubled (i.e. each initial report for bundle $x$ we have $\overline{v}(x) - v(x)$ and $v(x)-\underline{v}(x)$ is in expectation twice as large as in the standard bound generation procedure.

\item[Maximum Bounds (0,Inf)] The deviating bidders initial bounds for each bundle is supposed to be $[0,\infty]$. For our simulation we use 1.5 times social welfare to represent $\infty$ for numerical reasons.
\end{description}

The results are shown in Table~\ref{tbl:iMLCA-stability-mrvm} for MRVM (in the main text), Table~\ref{tbl:iMLCA-stability-gsvm} for GSVM (herein), and Table~\ref{tbl:iMLCA-stability-lsvm} for LSVM (herein). For each domain we tested deviations for all bidder types and report the resulting social welfare at truth, the marginal economy welfare for the deviating bidder at reported lower bounds and the utility of the deviating bidder.\footnote{The prices are computed based on reported lower bounds such that the utility cannot exactly be decomposed from the other two reported measures.}  

We observe very little difference in social welfare, marginal economy welfare or utility under any of the deviating strategies.  To formalize this, in the last row we show the p-value of the one-way ANOVA to test for statistically significant differences across all applied strategies for each column respectively.

Focusing on utility, we observe very high p-values (in all cases $>0.637$) for the utility in all domains and therefore can conclude that it is very unlikely that a bidder can improve their utility by using these strategies. A similar argument follows for social welfare and marginal economy welfare.

\begin{table}[H]
\centering
\resizebox{\textwidth}{!}{%
\begin{tabular}{l||p{1.7cm}p{1.7cm}p{1.7cm}|p{1.7cm}p{1.7cm}p{1.7cm}}
 & \multicolumn{3}{c|}{Regional} & \multicolumn{3}{c}{National} \\
 \hhline{~------}
 & Social Welfare & Marginal Economy Welfare
(at Lower Bounds) &       Utility & Social Welfare & Marginal Economy Welfare
(at Lower Bounds) &      Utility \\
\hhline{=::======}
Standard Bounds (SD=0.5) &    437.5 (3.6) &                     412.1 (3.7) &  25.19 (2.25) &    437.5 (3.6) &                     432.4 (4.1) &  4.78 (1.35) \\
\hhline{-||------}
Exact Values (SD=0)      &    437.5 (3.6) &                     411.8 (3.7) &  25.48 (2.29) &    437.5 (3.6) &                     431.7 (4.2) &  5.47 (1.62) \\
\hhline{-||------}
Wider Bounds (SD=1)      &    437.4 (3.6) &                     412.6 (3.6) &  24.55 (2.19) &    437.1 (3.6) &                     432.5 (4.1) &  4.37 (1.41) \\
\hhline{-||------}
Maximum Bounds (0, Inf)  &    437.4 (3.6) &                     414.6 (3.6) &  22.56 (2.07) &    436.5 (3.7) &                     432.6 (4.1) &  3.65 (1.35) \\
\hhline{-||------}
ANOVA p-value                  &       1.000 &                        0.944 &      0.787 &       0.997 &                        0.999 &     0.841 \\
\end{tabular}
}
\caption{Results for different bound reporting strategies for one bidder of each reported type in GSVM. Standard errors in parentheses. Averages over 100 instances.}
\label{tbl:iMLCA-stability-gsvm}
\end{table}

\begin{table}[H]
\centering
\resizebox{\textwidth}{!}{%
\begin{tabular}{l||p{1.7cm}p{1.7cm}p{1.7cm}|p{1.7cm}p{1.7cm}p{1.7cm}}
 & \multicolumn{3}{c|}{Regional} & \multicolumn{3}{c}{National} \\
\hhline{~------}
{} & Social Welfare & Marginal Economy Welfare
(at Lower Bounds) &       Utility & Social Welfare & Marginal Economy Welfare
(at Lower Bounds) &       Utility \\
\hhline{=::======}
Standard Bounds (SD=0.5) &    533.5 (4.5) &                     520.6 (6.0) &  12.88 (2.93) &    533.5 (4.5) &                     505.4 (4.6) &  28.04 (3.89) \\
\hhline{-||------}
Exact Values (SD=0)      &    532.8 (4.5) &                     520.3 (6.1) &  12.46 (2.93) &    533.4 (4.5) &                     505.5 (4.8) &  27.89 (3.59) \\
\hhline{-||------}
Wider Bounds (SD=1)      &    534.1 (4.4) &                     520.4 (6.1) &  13.67 (3.32) &    533.7 (4.4) &                     507.2 (4.6) &  26.46 (3.66) \\
\hhline{-||------}
Maximum Bounds (0, Inf)  &    533.2 (4.5) &                     520.9 (6.0) &  12.35 (2.93) &    533.0 (4.5) &                     510.1 (4.3) &  22.81 (3.19) \\
\hhline{-||------}
ANOVA p-value                  &       0.998 &                        1.000 &      0.990 &       1.000 &                        0.875 &      0.714 \\
\end{tabular}
}
\caption{Results for different bound reporting strategies for one bidder of each reported type in LSVM. Standard errors in parentheses. Averages over 100 instances.}
\label{tbl:iMLCA-stability-lsvm}
\end{table}

\bibliographystyle{plainnat}
\bibliography{iMLCA}    

\begin{thebibliography}{43}
\providecommand{\natexlab}[1]{#1}
\providecommand{\url}[1]{\texttt{#1}}
\expandafter\ifx\csname urlstyle\endcsname\relax
  \providecommand{\doi}[1]{doi: #1}\else
  \providecommand{\doi}{doi: \begingroup \urlstyle{rm}\Url}\fi

\bibitem[Ausubel and Cramton(2011)]{ausubel2011auction}
Lawrence Ausubel and Peter Cramton.
\newblock Auction design for wind rights.
\newblock Technical report, Report to Bureau of Ocean Energy Management,
  Regulation and Enforcement, 2011.

\bibitem[Ausubel and Baranov(2017)]{ausubel2017practical}
Lawrence~M Ausubel and Oleg Baranov.
\newblock A practical guide to the combinatorial clock auction.
\newblock \emph{The Economic Journal}, 127\penalty0 (605):\penalty0 334--350,
  2017.

\bibitem[Beyeler et~al.(2021)Beyeler, Brero, Lubin, and
  Seuken]{BeyelerImlca2021}
Manuel Beyeler, Gianluca Brero, Benjamin Lubin, and Sven Seuken.
\newblock Machine learning-powered iterative combinatorial auctions with
  interval bidding.
\newblock In \emph{Proceedings of the 22nd ACM Conference on Economics and
  Computation}, page 136, July 2021.

\bibitem[Bichler and Paulsen(2018)]{Bichler2018PrincipalAgent}
Martin Bichler and Per Paulsen.
\newblock A principal-agent model of bidding firms in multi-unit auctions.
\newblock \emph{Games and Economic Behavior}, 111:\penalty0 20--40, 2018.

\bibitem[Bichler et~al.(2013)Bichler, Shabalin, and Wolf]{bichler2013core}
Martin Bichler, Pasha Shabalin, and J{\"u}rgen Wolf.
\newblock Do core-selecting combinatorial clock auctions always lead to high
  efficiency? {A}n experimental analysis of spectrum auction designs.
\newblock \emph{Experimental Economics}, 16\penalty0 (4):\penalty0 511--545,
  2013.

\bibitem[Bikhchandani et~al.(2002)Bikhchandani, Ostroy,
  et~al.]{bikhchandani2002package}
Sushil Bikhchandani, Joseph~M Ostroy, et~al.
\newblock The package assignment model.
\newblock \emph{Journal of Economic theory}, 107\penalty0 (2):\penalty0
  377--406, 2002.

\bibitem[Blum et~al.(2004)Blum, Jackson, Sandholm, and
  Zinkevich]{blum2004preference}
Avrim Blum, Jeffrey Jackson, Tuomas Sandholm, and Martin Zinkevich.
\newblock Preference elicitation and query learning.
\newblock \emph{Journal of Machine Learning Research}, 5\penalty0
  (Jun):\penalty0 649--667, 2004.

\bibitem[Brero and Lahaie(2018)]{brero2018bayesian}
Gianluca Brero and Sébastien Lahaie.
\newblock A bayesian clearing mechanism for combinatorial auctions.
\newblock In \emph{Proceedings of the 32nd AAAI Conference on Artificial
  Intelligence}, pages 941--948, 2018.

\bibitem[Brero et~al.(2018)Brero, Lubin, and Seuken]{brero2018combinatorial}
Gianluca Brero, Benjamin Lubin, and Sven Seuken.
\newblock {Combinatorial Auctions via Machine Learning-based Preference
  Elicitation.}
\newblock In \emph{{Proceedings of the 27th International Joint Conference on
  Artificial Intelligence}}, pages 128--136, 2018.

\bibitem[Brero et~al.(2019)Brero, Lahaie, and Seuken]{brero2019fast}
Gianluca Brero, Sébastien Lahaie, and Sven Seuken.
\newblock Fast iterative combinatorial auctions via bayesian learning.
\newblock In \emph{Proceedings of the 33rd AAAI Conference on Artificial
  Intelligence}, pages 1820--1828, Jul 2019.

\bibitem[Brero et~al.(2021{\natexlab{a}})Brero, Eden, Gerstgrasser, Parkes, and
  Rheingans-Yoo]{Brero2021RLForIndirectMechanisms}
Gianluca Brero, Alon Eden, Matthias Gerstgrasser, David~C. Parkes, and Duncan
  Rheingans-Yoo.
\newblock Reinforcement learning of simple indirect mechanisms.
\newblock In \emph{Proceedings of the 35th AAAI Conference of Artificial
  Intelligence}, Virtual, 2021{\natexlab{a}}.

\bibitem[Brero et~al.(2021{\natexlab{b}})Brero, Lubin, and
  Seuken]{brero2019machine}
Gianluca Brero, Benjamin Lubin, and Sven Seuken.
\newblock Machine learning-powered iterative combinatorial auctions.
\newblock \emph{arXiv preprint}, abs/1911.08042, 2021{\natexlab{b}}.

\bibitem[B\"{u}nz et~al.(2018)B\"{u}nz, Lubin, and
  Seuken]{Lubin2018ComputationalSearch}
Benedikt B\"{u}nz, Benjamin Lubin, and Sven Seuken.
\newblock Designing core-selecting payment rules: A computational search
  approach.
\newblock In \emph{Proceedings of the 19th ACM Conference on Economics and
  Computation}, Ithaca, NY, June 2018.

\bibitem[Cavallo et~al.(2005)Cavallo, Parkes, Juda, Kirsch, Kulesza, Lahaie,
  Lubin, Michael, and Shneidman]{cavallo2005tbbl}
Ruggiero Cavallo, David~C Parkes, Adam~I Juda, Adam Kirsch, Alex Kulesza,
  S{\'e}bastien Lahaie, Benjamin Lubin, Loizos Michael, and Jeffery Shneidman.
\newblock Tbbl: A tree-based bidding language for iterative combinatorial
  exchanges.
\newblock In \emph{Multidisciplinary Workshop on Advances in Preference
  Handling}, 2005.

\bibitem[Connolly et~al.(2018)Connolly, Lim, Mitchell, and
  Trivedi]{connolly2018fcc}
Michelle Connolly, Elizabeth Lim, Frances Mitchell, and Akshaya Trivedi.
\newblock The 2016 fcc broadcast incentive auction.
\newblock In \emph{The 46th Research Conference on Communication, Information
  and Internet Policy}, page~63, 2018.

\bibitem[Cramton(2013)]{cramton2013spectrum}
Peter Cramton.
\newblock Spectrum auction design.
\newblock \emph{Review of industrial organization}, 42\penalty0 (2):\penalty0
  161--190, 2013.

\bibitem[Cramton et~al.(2006)Cramton, Shoham, and Steinberg]{cramtonCAintro}
Peter Cramton, Yoav Shoham, and Richard Steinberg.
\newblock Introduction to combinatorial auctions.
\newblock In Peter Cramton, Yoav Shoham, and Richard Steinberg, editors,
  \emph{Combinatorial Auctions}, chapter~1. MIT Press, 2006.

\bibitem[Day and Cramton(2012)]{day2012quadratic}
Robert~W Day and Peter Cramton.
\newblock Quadratic core-selecting payment rules for combinatorial auctions.
\newblock \emph{Operations Research}, 60\penalty0 (3):\penalty0 588--603, 2012.

\bibitem[D{\"u}tting et~al.(2015)D{\"u}tting, Fischer, Jirapinyo, Lai, Lubin,
  and Parkes]{dutting2015payment}
Paul D{\"u}tting, Felix Fischer, Pichayut Jirapinyo, John~K Lai, Benjamin
  Lubin, and David~C Parkes.
\newblock Payment rules through discriminant-based classifiers.
\newblock \emph{ACM Transactions on Economics and Computation (TEAC)},
  3\penalty0 (1):\penalty0 1--41, 2015.

\bibitem[D{\"u}tting et~al.(2019)D{\"u}tting, Feng, Narasimhan, Parkes, and
  Ravindranath]{dutting2019optimal}
Paul D{\"u}tting, Zhe Feng, Harikrishna Narasimhan, David Parkes, and
  Sai~Srivatsa Ravindranath.
\newblock Optimal auctions through deep learning.
\newblock In \emph{International Conference on Machine Learning}, pages
  1706--1715. PMLR, 2019.

\bibitem[Estermann et~al.(2023)Estermann, Kramer, Wattenhofer, and
  Wang]{estermann2023deep}
Benjamin Estermann, Stefan Kramer, Roger Wattenhofer, and Ye~Wang.
\newblock Deep learning-powered iterative combinatorial auctions with active
  learning.
\newblock In \emph{Proceedings of the 22nd International Conference on
  Autonomous Agents and Multiagent Systems}, pages 2919--2921, 2023.

\bibitem[Garey and Johnson(1978)]{Garey1978StrongNPHardness}
Michael Garey and D.~S. Johnson.
\newblock {``Strong''} np-completeness results: Motivation, examples, and
  implications.
\newblock \emph{Journal of the ACM}, 25:\penalty0 499--508, 1978.

\bibitem[Goeree and Holt(2010)]{goeree2010hierarchical}
Jacob~K Goeree and Charles~A Holt.
\newblock Hierarchical package bidding: A paper \& pencil combinatorial
  auction.
\newblock \emph{Games and Economic Behavior}, 70\penalty0 (1):\penalty0
  146--169, 2010.

\bibitem[Goetzendorf et~al.(2015)Goetzendorf, Bichler, Shabalin, and
  Day]{Goetzendorf:2015}
Andor Goetzendorf, Martin Bichler, Pasha Shabalin, and Robert~W. Day.
\newblock {Compact Bid Languages and Core Pricing in Large Multi-item
  Auctions}.
\newblock \emph{Management Science}, 61(7):\penalty0 1684--1703, 2015.

\bibitem[Golowich et~al.(2018)Golowich, Narasimhan, and
  Parkes]{golowich2018deep}
Noah Golowich, Harikrishna Narasimhan, and David~C Parkes.
\newblock Deep learning for multi-facility location mechanism design.
\newblock In \emph{Proceedings of the 27th International Joint Conference on
  Artificial Intelligence}, pages 261--267, 2018.

\bibitem[Kwasnica et~al.(2005)Kwasnica, Ledyard, Porter, and
  DeMartini]{kwasnica2005new}
Anthony~M Kwasnica, John~O Ledyard, Dave Porter, and Christine DeMartini.
\newblock A new and improved design for multiobject iterative auctions.
\newblock \emph{Management science}, 51\penalty0 (3):\penalty0 419--434, 2005.

\bibitem[Lahaie and Parkes(2004)]{lahaie2004applying}
Sebastien~M Lahaie and David~C Parkes.
\newblock Applying learning algorithms to preference elicitation.
\newblock In \emph{Proceedings of the 5th ACM conference on Electronic
  commerce}, pages 180--188, 2004.

\bibitem[Lubin et~al.(2008)Lubin, Juda, Cavallo, Lahaie, Shneidman, and
  Parkes]{Lubin.2008}
B.~Lubin, A.~I. Juda, R.~Cavallo, S.~Lahaie, J.~Shneidman, and D.~C. Parkes.
\newblock Ice: An expressive iterative combinatorial exchange.
\newblock \emph{Journal of Artificial Intelligence Research}, 33:\penalty0
  33--77, 2008.

\bibitem[Mas-Colell et~al.(1995)Mas-Colell, Whinston, Green,
  et~al.]{mas1995microeconomic}
Andreu Mas-Colell, Michael~Dennis Whinston, Jerry~R Green, et~al.
\newblock \emph{Microeconomic theory}, volume~1.
\newblock Oxford university press New York, 1995.

\bibitem[Oren and Williams(1975)]{oren1975competitive}
Matthew~E Oren and Albert~C Williams.
\newblock On competitive bidding.
\newblock \emph{Operations Research}, 23\penalty0 (6):\penalty0 1072--1079,
  1975.

\bibitem[Parkes(2006)]{parkes2006}
David~C Parkes.
\newblock Iterative combinatorial auctions.
\newblock In Peter Cramton, Yoav Shoham, and Richard Steinberg, editors,
  \emph{Combinatorial Auctions}, chapter~2, pages 41--78. MIT Press, 2006.

\bibitem[Sandholm(2013)]{sandholm2013very}
Tuomas Sandholm.
\newblock Very-large-scale generalized combinatorial multi-attribute auctions:
  Lessons from conducting \$60 billion of sourcing.
\newblock In Nir Vulkan, Alvin~E Roth, and Zvika Neeman, editors, \emph{The
  Handbook of Market Design}, chapter~1, page 379–412. Oxford University
  Press, Oxford, UK, 2013.

\bibitem[Sandholm and Boutilier(2006)]{sandholm2006preference}
Tuomas Sandholm and Craig Boutilier.
\newblock {Preference Elicitation in Combinatorial Auctions}.
\newblock In Peter Cramton, Yoav Shoham, and Richard Steinberg, editors,
  \emph{Combinatorial Auctions}, chapter~10. The MIT Press, 2006.

\bibitem[Scheffel et~al.(2012)Scheffel, Ziegler, and
  Bichler]{scheffel2012impact}
Tobias Scheffel, Georg Ziegler, and Martin Bichler.
\newblock On the impact of package selection in combinatorial auctions: an
  experimental study in the context of spectrum auction design.
\newblock \emph{Experimental Economics}, 15\penalty0 (4):\penalty0 667--692,
  2012.

\bibitem[Shen et~al.(2019)Shen, Lahaie, and Leme]{Shen2019LearningToClear}
Weiran Shen, Sebastien Lahaie, and Renato~Paes Leme.
\newblock Learning to clear the market.
\newblock In \emph{Proceedings of the 36th International Conference on Machine
  Learning}, pages 5710--5718, Longbeach, CA, 2019.

\bibitem[Soumalias et~al.(2022)Soumalias, Zamanlooy, Weissteiner, and
  Seuken]{soumalias2022machine}
Ermis Soumalias, Behnoosh Zamanlooy, Jakob Weissteiner, and Sven Seuken.
\newblock Machine learning-powered course allocation.
\newblock ar{X}iv preprint ar{X}iv:2210.00954, 2022.

\bibitem[Soumalias et~al.(2024)Soumalias, Weissteiner, Heiss, and
  Seuken]{soumalias2024machine}
Ermis~Nikiforos Soumalias, Jakob Weissteiner, Jakob Heiss, and Sven Seuken.
\newblock Machine learning-powered combinatorial clock auction.
\newblock In \emph{Proceedings of the 39th AAAI Conference on Artificial
  Intelligence}, pages 9891--9900, 2024.

\bibitem[Tang(2017)]{Tang2017ReinforcementMD}
Pingzhong Tang.
\newblock Reinforcement mechanism design.
\newblock In \emph{Proceedings of the 26th International Joint Conference on
  Artificial Intelligence}, Melbourne, Australia, August 2017.

\bibitem[Weiss et~al.(2017)Weiss, Lubin, and Seuken]{weiss2017sats}
Michael Weiss, Benjamin Lubin, and Sven Seuken.
\newblock {SATS: A Universal Spectrum Auction Test Suite}.
\newblock In \emph{Proceedings of the Sixteenth Conference on Autonomous Agents
  and MultiAgent Systems}, pages 51--59, Sao Paulo, Brazil, 2017.

\bibitem[Weissteiner and Seuken(2020)]{weissteiner2020deep}
Jakob Weissteiner and Sven Seuken.
\newblock Deep learning-powered iterative combinatorial auctions.
\newblock In \emph{Proceedings of the Thirty-fourth AAAI Conference on
  Artificial Intelligence}, pages 2284--2293, New York, NY, February 2020.

\bibitem[Weissteiner et~al.(2021)Weissteiner, Wendler, Seuken, Lubin, and
  P\"{u}schel]{Weissteiner2021FourierCAs}
Jakob Weissteiner, Chris Wendler, Sven Seuken, Ben Lubin, and Markus
  P\"{u}schel.
\newblock Fourier analysis-based iterative combinatorial auctions.
\newblock \emph{arxiv preprint}, abs/2009.10749, 2021.

\bibitem[Weissteiner et~al.(2022)Weissteiner, Heiss, Siems, and
  Seuken]{weissteiner2022monotone}
Jakob Weissteiner, Jakob Heiss, Julien Siems, and Sven Seuken.
\newblock Monotone-value neural networks: Exploiting preference monotonicity in
  combinatorial assignment.
\newblock In \emph{Proceedings of the 31st International Joint Conference on
  Artificial Intelligence}, pages 541--548, 2022.

\bibitem[Weissteiner et~al.(2023)Weissteiner, Heiss, Siems, and
  Seuken]{weissteiner2023bayesian}
Jakob Weissteiner, Jakob Heiss, Julien Siems, and Sven Seuken.
\newblock Bayesian optimization-based combinatorial assignment.
\newblock In \emph{Proceedings of the 37th AAAI Conference on Artificial
  Intelligence}, pages 5858--5866, 2023.

\end{thebibliography}

\end{document}